\documentclass[journal]{IEEEtran}
\makeatletter
\def\ps@headings{%
\def\@oddhead{\mbox{}\scriptsize\rightmark \hfil \thepage}%
\def\@evenhead{\scriptsize\thepage \hfil \leftmark\mbox{}}%
\def\@oddfoot{}%
\def\@evenfoot{}}
\makeatother
\usepackage{graphicx}
\usepackage{amssymb}
\usepackage{epstopdf}
\usepackage{amsmath}
\usepackage{enumerate}
\usepackage{cite}
\usepackage{verbatim}
\usepackage{mathcomp}
\usepackage{supertabular}
\usepackage{longtable}
\usepackage{stmaryrd}
\usepackage{url}
\usepackage{mathrsfs}
\usepackage[usenames,dvipsnames]{pstricks}
\usepackage{epsfig}
\usepackage{pst-grad}
\usepackage{pst-plot}
\usepackage[singlelinecheck=false, font = small, justification = raggedright, skip = -35pt]{caption}[2008/04/01]
\newcommand{\E}{{\mathcal E}}

\newcommand{\G}{{\mathcal G}}
\newcommand{\HH}{{\mathcal H}}
\newcommand{\I}{{\mathcal I}}
\newcommand{\J}{{\mathcal J}}
\newcommand{\LL}{{\mathcal L}}

\newcommand{\X}{{\mathcal X}}
\newcommand{\Y}{{\mathcal Y}}
\newcommand{\V}{{\mathcal V}}
\newcommand{\R}{{\mathcal R}}
\newcommand{\sS}{{\mathcal S}}
\newcommand{\N}{{\mathcal N}}

\newcommand{\cC}{{\mathscr{C}}}

\newcommand{\bM}{{\boldsymbol M}}

\newcommand{\bH}{{\boldsymbol H}} 
\newcommand{\bHi}{{\boldsymbol H}^{(i)}} 

\newcommand{\bL}{{\boldsymbol L}}

\newcommand{\bB}{{\boldsymbol B}}

\newcommand{\supp}{{\sf supp}}\newcommand{\dist}{{\mathsf{d}}}
\newcommand{\weight}{{\mathsf{wt}}}
\newcommand{\spn}{{\mathsf{span}}}
\newcommand{\clspn}{{\mathsf{colspan}}}

\newcommand{\rank}{{\mathsf{rank}}}
\newcommand{\mr}{{\text{min-rank}}}
\newcommand{\mrt}{{\text{min-rank}_2}}
\newcommand{\mrq}{{\text{min-rank}_q}}

\newcommand{\define}{\stackrel{\mbox{\tiny $\triangle$}}{=}}

\newcommand{\bu}{{\boldsymbol u}}
\newcommand{\bv}{{\boldsymbol v}}

\newcommand{\by}{{\boldsymbol y}}

\newcommand{\bc}{{\boldsymbol c}}

\newcommand{\bO}{{\boldsymbol 0}}

\newcommand{\bff}{{\boldsymbol f}}
\newcommand{\bg}{{\boldsymbol{g}}}

\newcommand{\bx}{{\boldsymbol{x}}}
\newcommand{\bz}{{\boldsymbol{z}}}

\newcommand{\bG}{{\boldsymbol{G}}}
\newcommand{\be}{{\boldsymbol e}}

\newcommand{\bep}{{\boldsymbol \epsilon}}
\newcommand{\beph}{\hat{\boldsymbol \epsilon}}

\newcommand{\bbti}{{{\boldsymbol \beta}_i}}

\newcommand{\bvi}{{\boldsymbol v}_{i}}
\newcommand{\Na}{{N_q[\alpha(\HH),2\delta+1]}}
\newcommand{\Nk}{{N_q[\kappa_q(\HH),2\delta+1]}}
\newcommand{\NX}{\N_q[\HH,\delta]}
\newcommand{\aG}{{\alpha(\HH)}}
\newcommand{\kG}{{\kappa_q(\HH)}}

\newcommand{\IX}{{\I(q,\HH)}}
\newcommand{\JX}{{\J(\HH)}}

\newcommand{\bxh}{{\hat{\boldsymbol{x}}}}

\newcommand{\fkD}{{\mathfrak D}}
\newcommand{\fkE}{{\mathfrak E}}
\newcommand{\fkJ}{{\mathfrak J}}

\newcommand{\al}{\alpha}

\newcommand{\kp}{\kappa}

\newcommand{\Ga}{\Gamma}
\newcommand{\Gnr}{\Gamma(n,\rho)}

\newcommand{\seq}{\subseteq}
\newcommand{\mic}{{$\HH$-IC}}
\newcommand{\dd}{{{$(\delta,\HH)$-ECIC}}}

\newcommand{\ra}{\rightarrow}
\renewcommand{\ge}{\geqslant}
\renewcommand{\le}{\leqslant}

\newcommand{\ff}{\mathbb{F}}
\newcommand{\fq}{\mathbb{F}_q}
\newcommand{\fqn}{\mathbb{F}_q^n}
\newcommand{\ft}{\mathbb{F}_2}
\newcommand{\et}{{\emph{et. al.}}}

\newcommand{\iss}{{(m^*,n,\X^*,f^*)}}
\newcommand{\issn}{{m^*,n,\X^*,f^*}}
\newcommand{\Hs}{{\HH^*}}

\newtheorem{theorem}{Theorem}[section]
\newtheorem{proposition}[theorem]{Proposition}
\newtheorem{lemma}[theorem]{Lemma}

\newtheorem{corollary}[theorem]{Corollary}
\newtheorem{definition}[theorem]{Definition}
\newtheorem{example}[theorem]{Example}
\newtheorem{remark}[theorem]{Remark}

\title{Error Correction for Index Coding with Side Information}

\begin{document}

\author{
  \IEEEauthorblockN{Son Hoang Dau$^*$, Vitaly Skachek$^{\dagger,1}$, and Yeow Meng Chee$^\ddagger$ \vspace{1ex}} \\
  \IEEEauthorblockA{$^{*,\ddagger}$Division of Mathematical Sciences,
    School of Physical and Mathematical Sciences\\
    Nanyang Technological University,
    21 Nanyang Link, Singapore 637371 \vspace{1ex}\\
    $^{\dagger}$Coordinated Science Laboratory, University of Illinois at Urbana-Champaign \\
    1308 W. Main Street, Urbana, IL 61801, USA \vspace{1ex} \\
    Emails: {\it $^{*}$daus0002@ntu.edu.sg, $^{\dagger}$vitalys@illinois.edu, $^{\ddagger}$YMChee@ntu.edu.sg }
  }
}

\maketitle

\footnotetext[1]{The work of this author was done while he was with the Division of Mathematical Sciences, School of Physical and Mathematical Sciences, Nanyang Technological University, 21 Nanyang Link, Singapore 637371.

A part of this work is to be presented in the \emph{IEEE International Symposium on Information Theory (ISIT)}, St. Petersburg, Russia, July-August 2011.}

\begin{abstract}
\boldmath
A problem of index coding with side information was first considered by Y. Birk and T. Kol \emph{(IEEE INFOCOM, 1998)}.
In the present work, a generalization of index coding scheme, where transmitted symbols are subject to errors, is studied. 
Error-correcting methods for such a scheme, and their parameters, are investigated. 
In particular, the following question is discussed:
given the side information hypergraph of index coding scheme and the maximal number of erroneous symbols $\delta$, 
what is the shortest length of a linear index code, such that every receiver is able to recover the required information? 
This question turns out to be a generalization of the problem of finding a shortest-length 
error-correcting code with a prescribed error-correcting capability 
in the classical coding theory. 

The Singleton bound and two other bounds, referred to as the $\al$-bound and
the $\kp$-bound, for the optimal length of a linear error-correcting index code (ECIC)
are established. For large alphabets, a construction based on concatenation of an optimal index 
code with an MDS classical code, is shown to attain the Singleton bound. 
For smaller alphabets, however, this construction may not be optimal. 
A random construction is also analyzed. It yields another inexplicit 
bound on the length of an optimal linear ECIC. 

Further, the problem of error-correcting decoding by a linear ECIC is studied. 
It is shown that in order to decode correctly the desired symbol, 
the decoder is required to find one of the vectors, belonging to an 
affine space containing the actual error vector. The syndrome decoding is shown to produce the correct 
output if the weight of the error pattern is less or equal to the error-correcting capability of the 
corresponding ECIC.  

Finally, the notion of static ECIC, which is suitable for use with a family of instances of an index coding problem, 
is introduced. Several bounds on the length of static ECIC's are derived, 
and constructions for static ECIC's are discussed. Connections of these codes to weakly   
resilient Boolean functions are established. 
\end{abstract}

\section{Introduction}
\label{sec:introduction}

\subsection{Background}

\PARstart{T}he problem of Index Coding with Side Information (ICSI) was introduced by Birk and Kol \cite{BirkKol98}, \cite{BirkKol2006}. 
During the transmission, each client might miss a certain part of the data, due to intermittent reception, limited storage capacity or any other reasons. Via a slow backward channel, the clients let the server know which messages they already have in their possession, and which messages they are interested to receive. The server has to find a way to deliver to each client all the messages he requested, yet spending a minimum number of transmissions. As it was shown in \cite{BirkKol98}, the server can significantly reduce the number of transmissions by coding the messages. 

The toy example in Figure 1 presents a scenario with one broadcast transmitter and four receivers. Each receiver requires 
a different information packet (we sometimes simply call it message). 
The na\"ive approach requires four separate transmissions, one transmission per an information
packet. However, by exploiting the knowledge on the subsets of messages that clients already have, and by using coding of the transmitted data, the server can just broadcast one coded packet. 

\vspace{-45pt}
\begin{figure}[h]
\begin{flushright}
\scalebox{1} 
{
\begin{pspicture}(0,-4.9003124)(7.7978125,5.0596876)
\definecolor{color1106b}{rgb}{0.8,0.8,0.8}
\pscircle[linewidth=0.04,dimen=outer,fillstyle=solid,fillcolor=color1106b](3.3278124,0.0496875){0.55}
\pscircle[linewidth=0.04,dimen=outer,fillstyle=solid,fillcolor=color1106b](1.3678125,1.9896874){0.51}
\pscircle[linewidth=0.04,dimen=outer,fillstyle=solid,fillcolor=color1106b](1.3278126,-1.9703125){0.51}
\pscircle[linewidth=0.04,dimen=outer,fillstyle=solid,fillcolor=color1106b](5.3478127,2.0096874){0.51}
\pscircle[linewidth=0.04,dimen=outer,fillstyle=solid,fillcolor=color1106b](5.3278127,-1.9303125){0.51}
\usefont{T1}{ptm}{m}{n}
\rput(3.3552186,0.0296875){$S$}
\usefont{T1}{ptm}{m}{n}
\rput(5.315219,-1.9303125){$R_3$}
\usefont{T1}{ptm}{m}{n}
\rput(5.315219,2.0096874){$R_4$}
\usefont{T1}{ptm}{m}{n}
\rput(1.3552188,1.9896874){$R_1$}
\usefont{T1}{ptm}{m}{n}
\rput(1.3152188,-1.9703125){$R_2$}
\pscircle[linewidth=0.022,linestyle=dashed,dash=0.16cm 0.16cm,dimen=outer](3.3578124,0.0196875){1.04}
\pscircle[linewidth=0.018,linestyle=dashed,dash=0.16cm 0.16cm,dimen=outer](3.3478124,0.0496875){1.75}
\pscircle[linewidth=0.02,linestyle=dashed,dash=0.16cm 0.16cm,dimen=outer](3.3378124,0.0396875){2.56}
\usefont{T1}{ptm}{m}{n}
\rput(4.8192186,0.0696875){$\sum_{i = 1}^4 x_i$}
\usefont{T1}{ptm}{m}{n}
\rput(5.777656,3.2296875){has $x_1,x_2,x_3$}
\usefont{T1}{ptm}{m}{n}
\rput(5.6076565,2.7896874){requests $x_4$}
\usefont{T1}{ptm}{m}{n}
\rput(1.4176563,3.2496874){has $x_2,x_3,x_4$}
\usefont{T1}{ptm}{m}{n}
\rput(1.2276562,2.7896874){requests $x_1$}
\usefont{T1}{ptm}{m}{n}
\rput(1.3776562,-2.7503126){has $x_1,x_3,x_4$}
\usefont{T1}{ptm}{m}{n}
\rput(1.2276562,-3.1903124){requests $x_2$}
\usefont{T1}{ptm}{m}{n}
\rput(5.777656,-2.7503126){has $x_1,x_2,x_4$}
\usefont{T1}{ptm}{m}{n}
\rput(5.6076565,-3.2103126){requests $x_3$}
\end{pspicture} 
}
\end{flushright}
\caption{An example of the ICSI problem}
\end{figure}
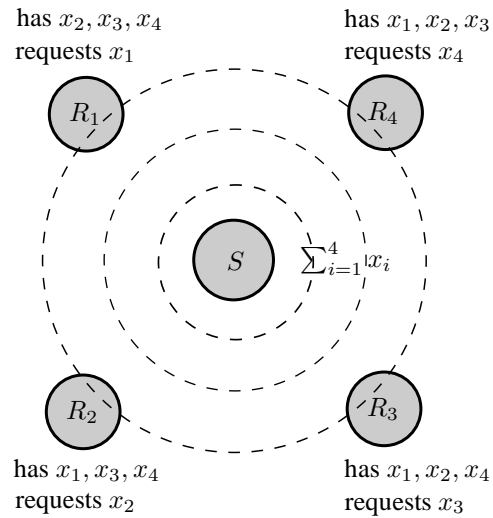

Possible applications of index coding include communications scenarios, in which a satellite or a server broadcasts a set of messages to a set clients, such as daily newspaper delivery or video-on-demand. Index coding with side information can also be used in opportunistic wireless networks. These are the networks in which a wireless node can opportunistically listen to the wireless channel. The client may obtain packets that are not designated to it (see \cite{Rouayheb2009, Katti2006, Katti2008}). As a result, a node obtains some side information about the transmitted data. Exploiting this additional knowledge may help to increase the throughput of the system. 

The ICSI problem has been a subject of several recent studies \cite{Yossef, Yossef-journal, LubetzkyStav, Wu, Rouayheb2007, Rouayheb2008, Rouayheb2009, ChaudhrySprintson, Alon}. This problem can be viewed as a special case of the Network Coding (NC) problem~\cite{Ahlswede},~\cite{KoetterMedard2003}. In particular, as it was shown in~\cite{Rouayheb2008, Rouayheb2009}, every instance of the NC problem can be reduced to an instance of the ICSI problem.

\subsection{Our contribution}

The preceding works on the ICSI problem consider scenario where the transmissions are error-free. 
In practice, of course, this might not be the case. In this work, we assume that the transmitted symbols 
are subject to errors. We extend some known results on index coding to a case where any receiver 
can correct up to a certain number of errors. It turns out that the problem of designing such error-correcting 
index codes (ECIC's) naturally generalizes the problem of constructing classical error-correcting codes. 

More specifically, assume that the number of messages that the server possesses is $n$, and that the designed
maximal number of errors is $\delta$. We show that the problem of constructing ECIC of minimal possible length 
is equivalent to the problem of constructing a matrix $\bL$ which has $n$ rows and the minimal possible 
number of columns, such that
\[
\weight\left(\bz \bL\right) \geq 2\delta + 1 \text{ for all } \bz \in \I, 
\]
where $\I$ is a certain subset of $\fq^n \backslash \{ \bO \}$. Here $\weight(\bx)$ denotes the Hamming weight of the vector $\bx$, $\fq$ stands for a finite field with $q$ elements, 
and $\bO$ is the all-zeros vector. 
If $\I = \fqn \backslash \{ \bO \}$, this problem becomes equivalent to the problem of
designing a shortest-length linear code of given dimension and minimum distance.  

In this work, we establish an upper bound (the $\kp$-bound) and a lower bound (the $\al$-bound) 
on the shortest length of a linear ECIC, which is able to correct any error pattern of size up to $\delta$. 
More specifically, let $\HH$ be the side information hypergraph that describes the instance of the ICSI problem.   
Let $\NX$ denote the length of 
a shortest-length linear ECIC over $\fq$, such that every $R_i$ can recover the desired message,
if the number of errors is at most $\delta$. We use notation $N_q[k,d]$ for the length of an optimal
linear error-correcting code of dimension $k$ and minimum distance $d$ over $\fq$. We obtain 
\begin{equation}
\label{EE1} 
\Na \leq \NX \leq \Nk,
\end{equation} 
where $\aG$ is the generalized independence number and $\kG$ is the min-rank (over $\fq$) of $\HH$. 

For linear index codes, we also derive an analog of the Singleton bound. This result  
implies that (over sufficiently large alphabet) the concatenation of a standard MDS error-correcting code 
with an optimal linear index code yields an optimal linear error-correcting index code. Finally, we consider 
random ECIC's. By analyzing its parameters, we obtain an upper bound on its length. 

When the side information hypergraph is a pentagon, and $\delta = 2$, 
the inequalities in~(\ref{EE1}) are shown to be strict. This implies that a concatenated scheme based on 
a classical error-correcting code and on a linear non-error-correcting index code does not necessarily yield 
an optimal linear error-correcting index code. Since ICSI problem can also be viewed as a source coding
problem~\cite{Yossef,Alon}, this example demonstrates that sometimes 
designing a single code for both source and channel coding can result in 
a smaller number of transmissions. 

The decoding of a linear ECIC is somewhat different from that of a classical error-correcting code.
There is no longer a need for a complete recovery of the whole information vector. 
We analyze the decoding criteria for the ECIC's and show that the syndrome decoding, which might be different for each 
receiver, results in a correct result, provided that the number of errors does not exceed
the error-correcting capability of the code.       

An ECIC is called static under a family of instances of the ICSI problem if it works for all of these instances. 
Such an ECIC is interesting since it remains useful as long as the parameters of the problem
vary within a particular range. Bounds and constructions for static ECIC's are studied in Section~\ref{sec:StaticECIC}. 
Connections between static ECIC's and weakly resilient vectorial Boolean functions are also discussed. 

The problem of error correction for NC was studied in several previous works. However, these results are not 
directly applicable to the ICSI problem. First, 
there is only a very limited variety of results for non-multicast networks in the existing literature. 
The ICSI problem, however, is a special case of the non-multicast 
NC problem. 
Second, the ICSI problem can be modeled by the NC scenario~\cite{Rouayheb2009}, 
yet, this requires that there are directed edges from particular sources to each sink, 
which provide the side information. The symbols transmitted on these special edges are not allowed to be corrupted. 
By contrast, for error-correcting NC, symbols transmitted on all edges can be corrupted. 

The paper is organized as follows. Basic notations and definitions, used throughout the paper, are provided in Section~\ref{sec:preliminaries}. The problem of index coding with and without error-correction is introduced in Section~\ref{sec:ic_ecc}. Some basic results are presented in that section. The $\alpha$-bound and the $\kappa$-bound
are derived in Section~\ref{subsec:bounds}. The Singleton bound is presented in Section~\ref{sec:singleton}. 
Random codes are discussed in Section~\ref{sec:random}. Syndrome decoding is studied in Section~\ref{subsec:decoding}. 
A notion of static error-correcting index codes is presented in Section~\ref{sec:StaticECIC}.
Several bounds on the length of such codes are derived,
and connections to resilient function are shown in that section. 
Finally, the results are summarized in Section~\ref{sec:conclusion}, and some open questions are proposed therein.   

\section{Preliminaries}
\label{sec:preliminaries}

In this section we introduce some useful notation. 
Here $\fq$ is the finite field of $q$ elements, where $q$ is a power of prime, and $\fq^*$ 
is the set of all nonzero elements of $\fq$. 

Let $[n] = \{1,2,\ldots,n\}$. For the vectors $\bu = (u_1, u_2, \ldots, u_n) \in \fq^n$ and $\bv = (v_1, v_2, \ldots, v_n) \in \fq^n$, 
the (Hamming) distance between $\bu$ and $\bv$ is defined to be the number of coordinates where $\bu$ and $\bv$ differ, 
namely, 
\[
\dist(\bu,\bv) = |\{i \in [n] \; : \; u_i \ne v_i\}| \; . 
\]
If $\bu \in \fq^n$ and $\bM \subseteq \fq^n$ is a set of vectors (or a vector subspace), 
then the last definition can be extended to 
\[
\dist(\bu,\bM) = \min_{\bv \in \bM} \dist(\bu, \bv) \; . 
\]

The \emph{support} of a vector $\bu \in \fqn$ is defined to be the set $\text{supp}(\bu) = \{i \in [n]: u_i \ne 0\}$.  
The (Hamming) weight of a vector $\bu$, denoted $\weight(\bu)$, is defined to be $|\text{supp}(\bu)|$, the number of nonzero coordinates of $\bu$. Suppose $E \subseteq [n]$. We write $\bu \lhd E$ whenever $\supp(\bu) \seq E$.  

A $k$-dimensional subspace $\cC$ of $\fq^n$ is called a linear $[n,k,d]_q$ code over $\fq$ if the minimum distance of $\cC$, 
\[
\dist(\cC) \define \min_{\bu \in \cC, \; \bv \in \cC, \; \bu \neq \bv} \dist(\bu,\bv) \; , 
\]
is equal to $d$. Sometimes we may use the notation $[n,k]_q$ for the sake of simplicity. The vectors in $\cC$ are called codewords. It is easy to see that the minimum weight of a nonzero codeword in a linear code $\cC$ is equal to its minimum distance $\dist(\cC)$. A \emph{generator matrix} $\bG$ of an $[n,k]_q$ code $\cC$ is a $k \times n$ matrix whose rows are linearly independent codewords of $\cC$. Then $\cC = \{\by \bG: \by \in \fq^k\}$. 
The \emph{parity-check matrix} of $\cC$ is an $(n-k) \times n$ matrix $\bH$ over $\fq$ such that 
$\bc \in \cC \Leftrightarrow \bH \bc^T = \bO^T$. 
Given $q$, $k$, and $d$, let $N_q[k,d]$ denote the length of the shortest linear code over $\fq$ which has dimension $k$ and minimum distance $d$.  

We use $\be_i = (\underbrace{0,\ldots,0}_{i-1},1,\underbrace{0,\ldots,0}_{n-i}) \in \fqn$ to denote the unit vector, which has a one at the $i$th position, and zeros elsewhere.
For a vector $\by = (y_1,y_2,\ldots,y_n)$ and a subset $B = \{i_1,i_2,\ldots,i_b\}$ of $[n]$, where $i_1 < i_2 < \cdots <i_b$, let $\by_B$ denote the vector $(y_{i_1},y_{i_2},\ldots,y_{i_b})$. 

For an $n \times N$ matrix $\bL$, let $\bL_i$ denote its $i$th row. For a set $E \subseteq [n]$, 
let $\bL_E$ denote the $|E| \times N$ matrix obtained from $\bL$ by deleting all the rows of $\bL$ which are not indexed by the elements of $E$. 
For a set of vectors $\bM$, we use notation $\spn(\bM)$ to denote the linear space spanned by the vectors in $\bM$.  
We also use notation $\clspn(\bL)$ for the linear space spanned by the columns of the matrix $\bL$. 

Let $\G = (\V, \E)$ be a graph with a vertex set $\V$ and an edge set $\E$.
The graph is called \emph{undirected} if every edge $e \in \E$, $e = \{ u, v \}$, and $u,v \in \V$. 
A graph $\G$ is \emph{directed} if every edge $e \in \E$ is an ordered pair $e = (u, v)$, 
$u, v \in \V$. A directed graph $\G$ is called \emph{symmetric} if 
\[
(u, v) \in \E \quad \Leftrightarrow \quad (v, u) \in \E \; . 
\]
There is a natural correspondence between undirected graph $\G = (\V, \E)$ and 
directed symmetric graph $\G' = (\V, \E')$ defined as 
\begin{equation}
\E = \left\{ \{ u, v \} \; : \; (u, v) \in \E' \right\} \; . 
\label{eq:correspond}
\end{equation}
Let $\G$ be an undirected graph. A subset of vertices $\sS \subseteq \V$ is called an \emph{independent set} if 
$\forall u, v \in \sS$, $\{u, v \} \notin \E$.
The size of the largest independent set in $\G$ is called the \emph{independence number} of $\G$,
and is denoted by $\al(\G)$. 
The graph $\bar{\G} = (\V, \bar{\E})$ is called the \emph{complement} of $\G = (\V, \E)$ if 
\[
\bar{\E} = \left\{ \{u, v\} \; : \; u \in \V, \, v \in \V, \, \{u,v\} \notin \E \right\} \; . 
\]
A \emph{coloring} of $\G$ using $\chi$ colors is a function
$\psi \; : \; \V \rightarrow [\chi]$, such that 
\[
\forall e = \{u, v\} \in \E \; : \; \psi(u) \neq \psi(v) \; . 
\]
The \emph{chromatic number} of $\G$ is the smallest number $\chi$ such that 
there exists a coloring of $\G$ using $\chi$ colors, and it is denoted by $\chi(\G)$.  
By using the correspondence~(\ref{eq:correspond}), the definitions of independence number, 
graph complement and chromatic number are trivially extended to directed symmetric graphs. 

\section{Index Coding and Error Correction}
\label{sec:ic_ecc}

\subsection{Index Coding with Side Information}
\label{subsec:icsi}

Index Coding with Side Information problem considers the following communications scenario. 
There is a unique sender (or source) $S$, who has a vector of messages 
$\bx = (x_1, x_2, \ldots, x_n)$ in his possession. 
There are also $m$ receivers $R_1,R_2,\ldots,R_m$, receiving information from $S$ via a broadcast channel. 
For each $i \in [m]$, $R_i$ has side information, 
i.e. $R_i$ owns a subset of messages $\{ x_j \}_{j \in \X_i}$, where $\X_i \subseteq [n]$. 
Each $R_i$, $i \in [m]$, 
is interested in receiving the message $x_{f(i)}$ (we say that $R_i$ requires $x_{f(i)}$), 
where the mapping $f: [m] \ra [n]$ satisfies $f(i) \notin \X_i$ for all $i \in [m]$. 
Hereafter, we use the notation $\X = (\X_1, \X_2, \ldots, \X_m)$. 
An instance of the ICSI problem is given by a quadruple $(m,n,\X,f)$. 
It can also be conveniently described by a directed hypergraph~\cite{Alon}. 

\vskip 10pt
\begin{definition}
Let $(m, n, \X, f)$ be an instance of the ICSI problem.  
The corresponding \emph{side information (directed) hypergraph} $\HH = \HH(m,n,\X,f)$ is defined by the vertex set
$\V = [n]$ and the edge set $\E_\HH$, where 
\[
\E_\HH = \{ (f(i), \X_i) \; : \; i \in [n]\} \; . 
\]
We often refer to $(m,n,\X,f)$ as an instance of the ICSI problem described by the hypergraph $\HH$. 
\end{definition}
\vskip 10pt 

Each side information hypergraph $\HH = (\V, \E_\HH)$ can be 
associated with the directed graph $\G_\HH = (\V,\E)$ in the following way. 
For each directed edge $(f(i), \X_i) \in \E_\HH$ there will be $|\X_i|$ directed edges $(f(i), v) \in \E$, for $v \in \X_i$. 
When $m = n$ and $f(i) = i$ for all $i \in [m]$, the graph $\G_\HH$ is, in fact, the \emph{side information graph}, 
defined in~\cite{Yossef}. 

The goal of the ICSI problem is to design a coding scheme that allows $S$ to satisfy 
the requests of all receivers $R_i$ in the least number of transmissions. More formally,
we have the following definition. 

\medskip
\begin{definition} 
An \emph{index code} over $\fq$ for an instance of the ICSI problem described by $\HH = \HH(m,n,\X,f)$  
(or just an {\mic} over $\fq$), is an encoding function
\begin{eqnarray*}
\fkE & : & \fq^n \rightarrow \fq^N \; , 
\end{eqnarray*}
such that for each receiver $R_i$, $i \in [m]$, there exists a decoding function
\[
\fkD_i \: : \: \fq^N \times \fq^{|\X_i|} \rightarrow \fq \; , \\
\]
satisfying
\[
\forall \bx \in \fq^n \; : \; \fkD_i(\fkE(\bx), \bx_{\X_i}) = x_{f(i)} \; .
\]
Sometimes we refer to such $\fkE$ as a \emph{non-error-correcting} index code. 
The parameter $N$ is called the \emph{length} of the index code. 
In the scheme corresponding to this code, 
$S$ broadcasts a vector $\fkE(\bx)$ of length $N$ over $\fq$.   
\label{def:IC}
\end{definition}
\medskip
\begin{definition} 
A \emph{linear index code} is an index code, for which the encoding function 
$\fkE$ is a linear transformation over $\fq$. 
Such a code can be described as 
\[
\forall \bx \in \fq^n \; : \; \fkE (\bx) = \bx \bL \; , 
\]
where $\bL$ is an $n \times N$ matrix over $\fq$. The matrix $\bL$ is called the \emph{matrix corresponding 
to the index code $\fkE$}. The code $\fkE$ is also referred to as the \emph{linear index code based on $\bL$}. 
\end{definition}
\vskip 10pt 

Hereafter, we assume that $\X = (\X_i)_{i \in [m]}$ is known to $S$.
Moreover, we also assume that the code $\fkE$ is known to each receiver $R_i$, $i \in [m]$. In practice 
this can be achieved by a preliminary communication session, when the knowledge of the sets 
$\X_i$ for $i \in [m]$ and of the code $\fkE$ are disseminated between the participants of 
the scheme. 

\medskip
\begin{definition}
Suppose $\HH = \HH(m,n,\X,f)$ corresponds to an instance of the ICSI problem. 
Then the \emph{min-rank} of $\HH$ over $\ff_q$ is defined as 
\begin{multline*}
\kp_q(\HH) \define \min\{\rank_{\fq}(\{\bvi + \be_{f(i)}\}_{i \in [m]}) \; : \; \\ \bvi \in \fqn \; , \; \bvi \lhd \X_i\} \; .
\end{multline*} 
\end{definition} 
\medskip
Observe that $\kp_q(\HH)$ generalizes the $\mr$ over $\fq$ of the side information graph, which was defined in~\cite{Yossef}.
More specifically, when $m = n$ and $f(i) = i$ for all $i \in [m]$, $\G_\HH$ becomes the side information graph, 
and $\kp_q(\HH) = \mrq(\G_\HH)$. 
The $\mr$ of an undirected graph was first introduced by Haemers~\cite{Haemers1978} to bound 
the Shannon capacity of a graph, and was later proved in~\cite{Yossef, Yossef-journal} to be the smallest number 
of transmissions in a linear index code. 

The following lemma was implicitly formulated in \cite{Yossef} for the case where $m = n$, $q = 2$, $f(i) = i$ for all $i \in [n]$, 
and generalized to its current form in~\cite{DauSkachekChee2010}.

\vskip 10pt 
\begin{lemma}
\label{lem:recovery}
Consider an instance of the ICSI problem described by $\HH = \HH(m,n,\X,f)$ . 
\begin{enumerate}
\item
The matrix $\bL$ corresponds to 
a linear {\mic} over $\fq$ if and only if for each $i \in [m]$ there exists $\bvi \in \fqn$ such that
\begin{itemize}
	\item $\bvi \lhd \X_i$ ; 
  \item $\bvi + \be_{f(i)} \in \clspn(\bL)$ . 
\end{itemize}     
\item The smallest possible length of a linear {\mic} over $\fq$ is $\kp_q(\HH)$. 
\end{enumerate}
\end{lemma}
\vskip 10pt 

\subsection{Error-Correcting Index Code with Side Information}
\label{subsec:ic_ecc_model}

Due to noise, the symbols received by $R_i$, $i \in [m]$, may be subject to errors.  
Consider an ICSI instance $(m,n,\X,f)$, and assume that $S$ broadcasts 
a vector $\fkE(\bx) \in \fq^N$. 
Let $\bep_i \in \fq^N$ be the error affecting the information received by $R_i$, $i \in [m]$. 
Then $R_i$ actually receives the vector
\[
\by_i = \fkE(\bx) + \bep_i \in \fq^N \; ,
\]
instead of $\fkE(\bx)$. The following definition is a generalization of Definition~\ref{def:IC}. 

\medskip
\begin{definition} 
Consider an instance of the ICSI problem described by $\HH = \HH(m,n,\X,f)$.
A \emph{$\delta$-error-correcting index code} ({\dd}) over $\fq$ for this instance is an encoding function
\begin{eqnarray*}
\fkE & : & \fq^n \rightarrow \fq^N \; , 
\end{eqnarray*}
such that for each receiver $R_i$, $i \in [m]$, there exists a decoding function
\[
\fkD_i \: : \: \fq^N \times \fq^{|\X_i|} \rightarrow \fq \; , \\
\]
satisfying
\[
\forall \bx, \bep_i \in \fq^n, \; \weight(\bep_i) \le \delta \; : \; \fkD_i(\fkE(\bx) + \bep_i, \bx_{\X_i}) = x_{f(i)}\; .
\]
\end{definition}
The definitions of the length, of a linear index code, and of the matrix corresponding 
to an index code are naturally extended to an error-correcting index code. 
Note that if $\fkE$ is an {\mic}, 
then it is a $(0,\HH)$-ECIC, and vice versa. 

\vskip 10pt
\begin{definition}
An \emph{optimal} linear {\dd} over $\fq$ is 
a linear {\dd} over $\fq$ of the smallest possible length $\NX$. 
\end{definition}
\vskip 10pt   

Consider an instance of the ICSI problem described by $\HH = \HH(m,n,\X,f)$. 
We define the set of vectors
\begin{multline*}
\IX \define \\ 
\left\{ \bz \in \fqn \; : \; \exists i \in [m] \text{ such that } \bz_{\X_i} = \bO \text{ and } z_{f(i)} \neq 0 \right\}.
\end{multline*}
For all $i \in [m]$, we also define 
\[
\Y_i \define [n] \backslash \Big( \{f(i)\} \cup \X_i \Big).
\] 
Then the collection of supports of all vectors in $\IX$ is given by
\begin{equation}
\label{Jdef} 
\JX \define \bigcup_{i \in [m]} \Big \{ \{f(i)\} \cup Y_i \; : \; Y_i \subseteq \Y_i \Big\}.
\end{equation} 
The necessary and sufficient condition for a matrix $\bL$ to be the matrix corresponding to some $(\delta, \HH)$-ECIC
is given in the following lemma.  

\vskip 10pt
\begin{lemma}
\label{lem:decodability}
The matrix $\bL$ corresponds to a $(\delta, \HH)$-ECIC over $\fq$ if and only if 
\begin{equation} 
\label{2:E9}
\weight\left(\bz \bL\right) \geq 2\delta+1 \text{ for all } \bz \in \IX \; . 
\end{equation} 
Equivalently, $\bL$ corresponds to a $(\delta, \HH)$-ECIC over $\fq$ if and only if
\[
\weight\left(\sum_{i \in K}z_i\bL_i\right) \geq 2\delta + 1,
\]
for all $K \in \JX$ and for all choices of $z_i \in \fq^*$, $i \in K$. 
\end{lemma}
\begin{proof}
For each $\bx \in \fq^n$, we define
\[
B(\bx,\delta) = \{\by \in \fq^N \; : \; \by = \bx \bL + \bep, \; \bep \in \fq^N , \; \weight(\bep) \leq \delta \} \; ,
\] 
the set of all vectors resulting from at most $\delta$ errors in the transmitted vector 
associated with the information vector $\bx$. 
Then the receiver $R_i$ can recover $x_{f(i)}$ correctly
if and only if 
\[
B(\bx,\delta) \cap B(\bx',\delta) = \varnothing,
\]
for every pair $\bx,\bx' \in \fqn$ satisfying:
\[
\bx_{\X_i} = \bx'_{\X_i} \text{ and } x_{f(i)} \neq x'_{f(i)} \; . 
\]
(Observe that $R_i$ is interested only in the bit $x_{f(i)}$, not in the whole vector $\bx$.)

Therefore, $\bL$ corresponds to a $(\delta, \HH)$-ECIC if and only if the following condition is satisfied: 
for all $i \in [m]$ and for all $\bx,\bx' \in \fqn$ such that
$\bx_{\X_i} = \bx'_{\X_i}$ and $x_{f(i)} \neq x'_{f(i)}$, it holds 
\begin{multline}
\forall \bep, \bep' \in \fq^N, \; \weight(\bep) \le \delta, \; \weight(\bep') \le \delta \; : \\
 \bx \bL + \bep \neq \bx' \bL + \bep' \; . 
\label{eq:unique-decode}
\end{multline}
Denote $\bz = \bx' - \bx$. Then, the condition in~(\ref{eq:unique-decode}) can be reformulated as follows: 
for all $i \in [n]$ and for all $\bz \in \fqn$ such that $\bz_{\X_i} = \bO$ and $z_{f(i)} \neq 0$, it holds
\begin{multline}
\forall \bep, \bep' \in \fq^N, \; \weight(\bep) \le \delta, \; \weight(\bep') \le \delta \; : \;
\bz \bL \neq \bep -\bep' \; . 
\label{eq:unique-decode-z}
\end{multline}
The equivalent condition is that for all $\bz \in \IX$, 
\[
\weight(\bz \bL) \ge 2 \delta + 1 \; . 
\]
Since for $\bz \in \IX$ we have
\[
\bz \bL = \sum_{i \in \supp(\bz)} z_i \bL_i,
\]
the condition (\ref{2:E9}) can be restated as
\[
\weight \left( \sum_{i \in K}z_i\bL_i \right) \geq 2 \delta + 1,
\]
for all $K \in \JX$ and for all choices of nonzero $z_i \in \fq$, $i \in K$. 
\end{proof}
\vskip 10pt

The next corollary follows from Lemma~\ref{lem:decodability} in a straight-forward manner. It is not hard to see that the conditions stated in Lemma~\ref {lem:decodability} and in the corollary below are, in fact, equivalent.

\medskip
\begin{corollary}
\label{cor:decodability}
For all $i \in [m]$, let
\[
\bM_i \define \mbox{span} \left(  \{ \bL_j \; : \; j \in \Y_i \} \right) \; . 
\]  
Then, the matrix $\bL$ corresponds to a $(\delta, \HH)$-ECIC over $\fq$ if and only if  
\begin{equation} 
\label{cor:decod}
\forall i \in [m] \; : \; \dist(\bL_{f(i)}, \bM_i) \ge 2 \delta + 1 \; .  
\end{equation} 
\end{corollary}
\medskip

The next corollary also follows directly from Lemma~\ref{lem:decodability} by considering an error-free
setup, i.e. $\delta = 0$. It is easy to verify that the conditions stated in this corollary and in Lemma~\ref{lem:recovery}
are equivalent, as expected. 
 
\vskip 10pt
\begin{corollary}
\label{coro:ic_decodability}
The matrix $\bL$ corresponds to an $\HH$-IC over $\fq$ if and only if 
\[
\weight\left(\sum_{i \in K}z_i\bL_i\right) \geq 1,
\]
for all $K \in \JX$ and for all choices of $z_i \in \fq^*$, $i \in K$, 
or, equivalently,  
\[
\forall i \in [m] \; : \; \bL_{f(i)} \notin \spn(\{\bL_j\}_{j \in \Y_i}) \; . 
\]
\end{corollary}

\medskip
\begin{example}
\label{example:1}
Let $q = 2$, $m = n=3$, and $f(i) = i$ for $i \in [3]$. 
Suppose $\X_1 = \{2,3\}$, $\X_2 = \{1,3\}$, and $\X_3 = \{1,2\}$. Let 
\[
\bL = \begin{pmatrix} 1 & 1 & 1 & 0 \\ 1 & 1 & 0 & 1\\ 1 & 0 & 1 & 1\end{pmatrix}. 
\]
Note that $\bL$ generates a $[4,3,1]_2$ code, which has minimum distance one. However, 
the index code based on $\bL$ can still correct one error. Indeed, let $\HH = \HH(3,3,\X,f)$, we have
\[
\I(2,\HH) = \{100,010,001\}. 
\]
Since each row of $\bL$ has weight at least three, it follows that 
$\weight(\bz \bL) \geq 3$ for all $\bz \in \I(2,\HH)$. By Lemma~\ref {lem:decodability}, 
$\bL$ corresponds to a $(1,\HH)$-ECIC over $\ff_2$. 

In fact, for this instance, even a simpler index code of length three, based on 
\[
\bL' = \begin{pmatrix} 1 & 1 & 1 \\ 1 & 1 & 1\\ 1 & 1 & 1\end{pmatrix} \; , 
\]
is a $(1,\HH)$-ECIC over $\ff_2$. 
\end{example} 
\vskip 10pt 

\begin{example}
Assume that $m = n$ and $f(i) = i$ for all $i \in [m]$. Furthermore, suppose that $\X_i = \varnothing$ 
for all $i \in [m]$ (i.e. there is no side information available to the receivers). Let $\HH = \HH(m,n,\X,f)$.
Then, $\IX = \fqn \backslash \{ \bO \}$.
Hence, by Lemma~\ref {lem:decodability}, the $n \times N$ matrix $\bL$ corresponding to a {\dd} over $\fq$ 
(for some integer $\delta \ge 0$) is a generating matrix of an $[N,n, \ge 2 \delta+1]_q$ linear code. 
Thus, under these settings, the problem of designing an optimal ECIC is reduced to the problem of constructing an 
optimal classical linear error-correcting code.
\end{example}
\medskip 

Observe however, that for general $\X$, changing the order of rows in $\bL$ can lead to ECIC's with different error-correcting capabilities. Therefore, the problem of designing an optimal linear ECIC is essentially the problem of finding the matrix $\bL$ corresponding to that code. However, the minimum distance 
of the code generated by the rows of $\bL$ is not necessary a valid indicator for goodness of an ECIC. Sometimes, as 
Example~\ref{example:1} shows, matrix $\bL$ with redundant rows yields a good ECIC.  
  
\section{The $\al$-Bound and the $\kp$-Bound}
\label{subsec:bounds}

Let $(m,n,\X,f)$ be an instance of the ICSI problem, and let $\HH$ be the corresponding side information hypergraph. 
Next, we introduce the following definitions for the hypergraph $\HH$. 

\vskip 10pt 
\begin{definition}
A subset $H$ of $[n]$ is called a \emph{generalized independent set}
in $\HH$ if every nonempty subset $K$ of $H$ belongs to $\JX$. 
\end{definition}
\vskip 10pt 

\begin{definition}
A generalized independent set of the largest size in $\HH$ is called a \emph{maximum generalized independent set}.
	The size of a maximum generalized independent set in $\HH$ is called the \emph{generalized independence number}, 
	and denoted by $\al(\HH)$. 
\end{definition}
\vskip 10pt 

When $m = n$ and $f(i) = i$ for all $i \in [n]$, the
generalized independence number of $\HH$ is equal to the maximum size of an acyclic induced subgraph of $\G_\HH$, which was introduced in~\cite{Yossef}. In particular, when $\G_\HH$ is symmetric, $\al(\HH)$ is the 
independence number of $\G_\HH$. We prove the latter statement in the Appendix.   

Next, we present a lower bound on the length of a $(\delta, \HH)$-ECIC. We call this bound \emph{$\al$-bound}.

\vskip 10pt
\begin{theorem}[$\al$-bound]
\label{thm:lowerbound} 
The length of an optimal linear $(\delta, \HH)$-ECIC over $\fq$ satisfies 
\[
\NX \ge \Na \; .
\]
Moreover, the equality is attained if there exists an $n \times \aG$ matrix $\bB = (b_{i,j})$
over $\fq$ satisfying the following condition:
for all $K \in \JX$ and for all choices of $z_i \in \fq^*$, $i \in K$, there always exists some $j$ such that
\[
	\sum_{i \in K} z_i \, b_{i,j} \neq 0 \; .
\]
\end{theorem}
\begin{proof}
Consider an $n \times N$ matrix $\bL$, which corresponds to a $(\delta, \HH)$-ECIC.
Let $H = \{i_1,i_2,\ldots,i_\aG\}$ be a maximum generalized independent set in $\HH$. 
Then, every subset $K \seq H$ satisfies $K \in \JX$. Therefore, 
\[
\weight\left( \sum_{i \in K} z_i \bL_i \right) \geq 2\delta + 1
\]
for all $K \subseteq H$, $K \neq \varnothing$, and for all choices of $z_i \in \fq^*$, $i \in K$. 
Hence, the $\aG$ rows of $\bL$, namely $\bL_{i_1},\bL_{i_2},\ldots,\bL_{i_\aG}$, form
a generator matrix of an $[N,\aG,2\delta+1]_q$ code. Therefore,
\[
N \geq \Na \; .
\]

Next, we assume the existence of a matrix $\bB$ satisfying the properties stated in the theorem. 
Let $\bL'$ be a generator matrix of some $[N',\aG,2\delta+1]_q$ code, where $N' = \Na$. 
We construct the $n \times N'$ matrix $\bL$ as follows. For $i \in [n]$, let
\[
\bL_i = \sum_{j = 1}^\aG b_{i,j}\bL'_j \; .
\] 
For every $K \in \JX$ and for all choices of $z_i \in \fq^*$, $i \in K$, we have
\[
\begin{split}
\weight \left( \sum_{i \in K} z_i \, \bL_i \right)  & 
= \weight \left( \sum_{i \in K} z_i \sum_{j = 1}^\aG b_{i,j} \bL'_j \right) \\
& = \weight \left( \sum_{j = 1}^\aG \left( \sum_{i \in K} z_i \, b_{i,j} \right) \bL'_j \right)\\
& \geq 2\delta + 1 \; ,
\end{split} 
\]
where the last transition is due to the existence of $j \in [\aG]$ such that 
\[
\sum_{i \in K} z_i \, b_{i,j} \neq 0 \; ,
\]
and the fact that $\bL'_j$'s are linearly independent nonzero codewords of a code of minimum distance $2\delta + 1$. 

We conclude that the index code based on $\bL$ is capable of correcting $\delta$ errors. Therefore, $\NX = \Na$. 
\end{proof}
\vskip 10pt

\begin{example}
\label{example:2}
Let $q = 2$, $m = n = 5$, $f(i) = i$ for all $i \in [m]$, and $\delta = 2$. Assume 
\begin{eqnarray*} 
\X_1 = \{2,3,4\} \; , \quad \X_2 = \{3,4,5\} \; , \quad \X_3 = \{4,5,1\} \; , \\
\X_4 = \{5,1,2\} \; , \quad \X_5 = \{1,2,3\} \; .
\end{eqnarray*} 
Let $\HH = \HH(5,5,\X,f)$. Then 
\[
\begin{split}
\JX \; = \; & \Big\{ \{1\},\{1,5\},\{2\},\{2,1\},\{3\},\\
&\{3,2\},\{4\},\{4,3\},\{5\},\{5,4\}\Big\}.
\end{split}
\] 
It is easy to check that $\aG = 2$. Therefore, Theorem~\ref{thm:lowerbound} implies that
\[
\N_2[\HH,2] \geq N_2[2,5] = 8 \; . 
\]
The last equality can be verified by~\cite{Grassl}.
 
On the other hand, take the matrix 
\[
\bB \define
\begin{pmatrix}
1 & 0 \\ 
0 & 1 \\
1 & 0 \\
0 & 1 \\
1 & 1 \\
\end{pmatrix} \; . 
\]
The matrix $\bB$ satisfies the property that for all $K \in \JX$, $K \neq \varnothing$, 
there exists $j \in [2]$ such that
\[
\sum_{i \in K} b_{i,j} \neq 0 \; . 
\]
From Theorem~\ref{thm:lowerbound}, we have $\N_2[\HH,2] = N_2[2,5] = 8$.
\end{example}
\vskip 10pt 

\begin{remark}
In~\cite{Yossef}, when $m = n$ and $f(i) = i$ for all $i \in [n]$, 
$\aG$ was shown to be a lower bound on the length of a (non-error-correcting) 
linear index code. However, the $\al$-bound in Theorem~\ref{thm:lowerbound} does not follow from 
the results in~\cite{Yossef}. The reason is that a concatenation of an optimal linear error-correcting 
code with an optimal non-error-correcting index code might fail to produce an optimal linear ECIC.  
This is illustrated later in Example~\ref{ex:1}. 
\end{remark}
\vskip 10pt 

The following proposition is based on the fact that concatenation of 
a $\delta$-error-correcting code with an optimal (non-error-correcting) $\HH$-IC yields a $(\delta, \HH)$-ECIC. 

\medskip
\begin{proposition}[$\kp$-bound]
\label{pro:upperbound}
The length of an optimal $(\delta, \HH)$-ECIC over $\fq$ satisfies
\[
\NX \leq \Nk \; .
\]
\end{proposition}
\begin{proof}
Let $\bG$, which is an $n \times \kG$ matrix, correspond to an optimal {\mic} over $\fq$. 
Denote 
\[
\by = \bx \bG \in \fq^\kG.
\] 
Let $\bM$ be a generator matrix of an optimal $[N,\kG,2\delta+1]_q$ code $\cC'$, where
\[
N = \Nk.
\] 
Consider a scheme where $S$ broadcasts the vector $\by \bM \in \fq^N$. If less than $\delta$ 
errors occur, then each receiver $R_i$ is able to recover $\by$ by using $\cC'$. Hence each $R_i$ is 
able to recover $x_{f(i)}$. Therefore, for the index code based on $\bL$, 
\[
\bL = \bG \bM \; , 
\]
each receiver $R_i$ is capable to recover $x_{f(i)}$ if the number of errors is less or equal to $\delta$. 
The length of the corresponding ECIC is 
$N = \Nk$. Therefore, 
\[
\NX \leq \Nk \; .
\]
\end{proof}
\vskip 10pt

By combining the results in Theorem~\ref{thm:lowerbound} and in Proposition~\ref{pro:upperbound}, we obtain the
following corollary. 

\vskip 10pt 
\begin{corollary}
\label{coro:sandwiched}
The length of an optimal linear $(\delta, \HH)$-ECIC over $\fq$ satisfies
\[
\Na \leq \NX \leq \Nk \; .
\] 
\end{corollary}
\vskip 10pt

It is shown in the example below that the inequalities in Corollary~\ref{coro:sandwiched}
can be strict. In particular, it follows that mere application of an error-correcting code on top of an 
index code may fail to provide us with an optimal linear ECIC. This fact  
motivates the study of ECIC's in Sections~\ref{sec:ic_ecc}--\ref{subsec:decoding}.    

\vskip 10pt 
\begin{example}
\label{ex:1}
Let $q = 2$, $m = n = 5$, $\delta = 2$, and $f(i) = i$ for all $i \in [m]$. Assume
\begin{eqnarray*} 
\X_1 = \{2,5\} \; , \quad 
\X_2 = \{1,3\} \; , \quad 
\X_3 = \{2,4\} \; , \\
\X_4 = \{3,5\} \; , \quad
\X_5 = \{1,4\} \; . 
\end{eqnarray*}
Let $\HH = \HH(5,5,\X,f)$. Then we have
\[
\begin{split}
\JX=\Big\{&\{1\},\{1,3\},\{1,4\},\{1,3,4\},\\
&\{2\},\{2,4\},\{2,5\},\{2,4,5\},\\
&\{3\},\{1,3\},\{3,5\},\{1,3,5\},\\
&\{4\},\{1,4\},\{2,4\},\{1,2,4\},\\
&\{5\},\{2,5\},\{3,5\},\{2,3,5\}
\Big\}.
\end{split}
\]
The side information graph $\G_\HH$ of this instance is a pentagon. It is easy to verify that $\aG = \al(\G_\HH) = 2$. 
It follows from Theorem~9 in~\cite{Yossef-journal} that $\kp_2(\HH) = \mrt(\G_\HH) =~3$. 
Thus, from~\cite{Grassl} we have 
\[
N_2[2,5] = 8 \quad \mbox{ and } \quad N_2[3,5] = 10 \; . 
\]
Due to Corollary~\ref{coro:sandwiched}, we have
\[
8 \leq \N_2[\HH,2] \leq 10. 
\]
Using a computer search, we obtain that $\N_2[\HH,2] = 9$, and the corresponding optimal scheme is based on
\[
\bL = \begin{pmatrix} 
1  &1 & 1  &1 & 1&  0&  0 & 0  &0  \\
0  &1 & 0  &1  &1  &0  &1 & 1 & 0  \\
1  &1  &0  &0  &0  &1  &1  &1  &0  \\
0  &1  &1  &0  &0  &1  &0  &1  &1  \\
1  &0  &1  &0  &1  &0  &0  &1  &1  \\
\end{pmatrix} \; . 
\]
It is technical to verify that for all $K \in \JX$,  
\[
\weight \left( \sum_{i \in K} \bL_i \right) \geq 5 \; .
\]
Therefore by Lemma~\ref {lem:decodability}, for the index code based on $\bL$, each receiver $R_i$ is 
able to recover $x_i$, if the number of errors is less than or equal to $2$. 
Observe that the length of the ECIC corresponding to $\bL$ lies strictly between the $\al$-bound and the $\kp$-bound. 
\end{example}
\vskip 10pt 

When the graph $\G$ is undirected (or symmetric), 
the following theorem holds (see, for instance,~\cite{Haemers1978}).

\vskip 10pt 
\begin{theorem}
Let $\chi(\bar{\G)}$ denote the chromatic number of the complement of the graph $\G$. Then,
\[
\al(\G) \leq \mrq(\G) \leq \chi(\bar{\G)} \; . 
\]
\end{theorem}
\medskip

When $m = n$ and $f(i) = i$ for all $i \in [m]$, we have that $\al(\HH) = \al(\G_\HH)$ and $\kG = \mrq(\G_\HH)$. Moreover, if the graph $\G_\HH$ is symmetric and satisfies $\al(\G_\HH) = \chi(\bar{\G}_\HH)$,
then from Corollary~\ref{coro:sandwiched} we have
\[
\NX = \Na = \Nk \; ,
\]
for all $q$, and the corresponding bounds in Corollary~\ref{coro:sandwiched} are tight. 
\medskip

\begin{definition}
An undirected (or symmetric) graph $\G$ is called perfect if for every induced subgraph $\G'$ of 
$\G$, $\al(\G') = \chi(\bar{\G'})$. 
\end{definition}
\medskip

Perfect graphs include families of graphs such as trees, bipartite graphs, interval graphs, 
and chordal graphs. If $m = n$, $f(i) = i$ for all $i \in [m]$, and $\G_\HH$ is perfect, then
the bounds in Corollary~\ref{coro:sandwiched} are tight.
For the full characterization of perfect graphs, the reader can refer to~\cite{Chudnovsky}. 
 
\section{The Singleton Bound}
\label{sec:singleton}

The following bound is analogous to Singleton bound for classical linear error-correcting codes. 
\medskip
\begin{theorem}[Singleton bound]
\label{thrm:singleton}
The length of an optimal  linear {\dd} over $\fq$
satisfies
\[
\NX \geq \kG + 2 \delta \; .
\]
\end{theorem}
\begin{proof}
Let $\bL$ be the $n \times \NX$ matrix corresponding to some optimal {\dd}. 
Let $\bL'$ be the matrix obtained by deleting any $2 \delta$ 
columns from $\bL$. 

By Lemma~\ref{lem:decodability}, $\bL$ satisfies
\[
\weight\left(\sum_{i \in K} z_i \bL_i\right) \geq 2 \delta + 1 \; ,
\]
for all $K \in \JX$ and all choices of $z_i \in \fq^*$, $i \in K$. 
We deduce that the rows of $\bL'$ also satisfy
\[
\weight \left(\sum_{i \in K} z_i \bL'_i \right) \geq 1 \; .
\] 
By Corollary~\ref{coro:ic_decodability}, 
$\bL'$ corresponds to a linear {\mic}.  
Therefore, by Lemma~\ref {lem:recovery}, part 2, $\bL'$ has at least $\kG$ columns. We deduce that
\[
\NX - 2 \delta \geq \kG \; ,
\]
which concludes the proof. 
\end{proof}
\vskip 10pt

The following corollary from Proposition~\ref{pro:upperbound} and Theorem~\ref{thrm:singleton} demonstrates  
that, for sufficiently large alphabets, a concatenation of a classical MDS error-correcting code
with an optimal (non-error-correcting) index code yields an optimal ECIC. 
However, as it was illustrated in Example~\ref{ex:1}, this does not hold for the index coding 
schemes over small alphabets.

\vskip 10pt
\begin{corollary} [MDS error-correcting index code]
For $q \geq \kG + 2 \delta - 1$, 
\begin{equation}
\NX = \kG + 2 \delta \; . 
\label{eq:mds}
\end{equation}
\end{corollary}
\begin{proof}
From Theorem~\ref{thrm:singleton}, we have 
\[
\NX \geq \kG + 2 \delta \; . 
\]
On the other hand, from Proposition~\ref{pro:upperbound}, 
\[
\NX \leq N_q[\kG, 2\delta+1] = \kG + 2 \delta \; ,
\]
for $q \geq \kG + 2 \delta - 1$ (by taking doubly-extended Reed-Solomon (RS) codes). Therefore, for these $q$,~(\ref{eq:mds}) holds.
\end{proof}

\medskip
\begin{remark}
Let $q = 2$, $m = n = 2\ell + 1$ ($\ell \geq 2$), and $f(i) = i$ for all $i \in [m]$. 
Let $\X_1 = \{2,n\}$ and $\X_n = \{1,n-1\}$. 
For $2 \leq i \leq n$, let $\X_i = \{i-1,i+1\}$. Let $\HH = \HH(n,n,\X,f)$. 
Then $\G_\HH$ is the (symmetric) odd cycle of length $n$. 
Therefore, $\al(\HH) = \al(\G_\HH) = \ell$. From~\cite{Yossef-journal}, $\kp_2(\HH) = \mrt(\G_\HH) = \ell + 1$. 
From $\al$-bound,
\[
	\N_2[\HH,\delta] \geq N_2[\ell,2\delta + 1] \; .
\]
By contrast, from Theorem~\ref{thrm:singleton},
\[
	\N_2[\HH,\delta] \geq (\ell + 1) + 2\delta \; .
\]
As there are no nontrivial binary MDS codes, we have
\[
N_2[\ell,2\delta + 1] \geq \ell + 2\delta + 1 \; ,
\]	
for all choices of $\delta>0$. 
Therefore, for these choices, the $\al$-bound is at least as good as the Singleton bound. 
\end{remark}

\section{Random codes}
\label{sec:random}

In this section we prove an inexplicit upper bound on the optimal length of the ECIC's. The proof is based on constructing a random ECIC and analyzing its parameters.  

\medskip
\begin{theorem}
\label{thrm:random}
Let $\HH = \HH(m,n,\X,f)$ describe an instance of the ICSI problem. 
Then there exists a {\dd} over $\fq$ of length $N$ if
\[
\sum_{i \in [m]} q^{n - |\X_i| - 1} < \frac{q^N}{V_q(N, 2\delta)} \; , 
\]
where    
\begin{equation}
\label{eq:sphere}
V_q(N, 2\delta) = \sum_{\ell = 0}^{2 \delta} {N \choose \ell} (q-1)^\ell
\end{equation}
is the volume of the $q$-ary sphere in $\fq^N$. 

\end{theorem}
\begin{proof}
We construct a random $n \times N$ matrix $\bL$ over $\fq$, row by row. 
Each row is selected independently of other rows, uniformly over $\fq^N$.
Define vector spaces 
\[
\bM_i \define \mbox{span} \left(  \{ \bL_j \; : \; j \in \Y_i \} \right) 
\]  
for all $i \in [m]$.
We also define the following events:
\begin{eqnarray*}
\forall i \in [m] \; : \; \mbox{ Event $E_i$} \; \define \; \left\{ \dist( \bL_{f(i)}, \bM_i ) < 2 \delta + 1 \right\} \; ,
\end{eqnarray*}
and 
\begin{multline*}
\mbox{Event $E_{{Fail}}$} \; \define \\
\; \left\{ \mbox{$\bL$ does not correspond to a {\dd}} \right\} \; .    
\end{multline*}

The event $E_i$ represents the situation when the receiver $R_i$ cannot recover $x_{f(i)}$. 
Then, by Corollary~\ref{cor:decodability}, the event $E_{{Fail}}$ is equivalent to $\bigcup_{i \in [m]} E_i$. Therefore, 
\begin{equation}
\Pr \left( E_{{Fail}} \right) = \Pr \left( \bigcup_{i \in [m]} E_i \right) \le \sum_{i \in [m]} \Pr \left( E_i \right) \; . 
\label{eq:failure}
\end{equation} 
For a particular event $E_i$, $i \in [m]$, 
\begin{equation}
\Pr \left( E_i \right) \le \frac{q^{|\Y_i|} \, V_q(N, 2\delta)}{q^N} \; . 
\label{eq:e_i}
\end{equation}

There exists a matrix $\bL$ that corresponds to a {\dd} if $\Pr \left( E_{{Fail}} \right) < 1$. It is enough to require 
that the right-hand side of~(\ref{eq:failure}) is smaller than $1$. By plugging in the expression in~(\ref{eq:e_i}), 
we obtain a sufficient condition on the existence of a {\dd} over $\fq$: 
\[
\frac{V_q(N, 2\delta)}{q^N} \, \sum_{i \in [m]} q^{|\Y_i|} < 1 \; .
\]
\end{proof}
\vskip 10pt 

\begin{remark}
The bound in Theorem~\ref{thrm:random} does not take into account 
the structure of the sets $\X_i$'s, other than their cardinalities.
Therefore, this bound generally is weaker than
the $\kappa$-bound. On the other hand, for a particular instance of 
the ICSI problem, it is easier to compute this bound, while calculating
the $\kappa$-bound in general is an NP-hard problem.
\end{remark}
\medskip

\begin{remark}
The bound in Theorem~\ref{thrm:random} implies a bound on $\kp_q(\HH)$, which is tight for some $\HH$. 
Indeed, fix $\delta = 0$. The bound implies that there exists a linear index code of length $N$ whenever 
\begin{equation}
\sum_{i \in [m]} q^{n - |\X_i| - 1} < q^N \; . 
\label{eq:kappa-bound}
\end{equation}
Let $m = n = 2\ell + 1$ ($\ell \geq 2$), and $f(i) = i$ for all $i \in [n]$. 
Let $\X_1 = [n] \backslash \{1,2,n\}$ and $\X_n = [n] \backslash \{1,n-1,n\}$. 
For $2 \leq i \leq n - 1$, let $\X_i = [n] \backslash \{i-1,i,i+1\}$. 
Let $\HH = \HH(n,n,\X,f)$ be the corresponding side information hypergraph.  
Then $\G_\HH$ is the complement of the (symmetric directed) odd cycle of length $n$. 
We have $|\X_i| = 2\ell  - 2$ for all $i \in [n]$. 
Then~(\ref{eq:kappa-bound}) becomes
\[
N > 2 + \log_q(2\ell + 1) \; . 
\]
If $q > 2 \ell + 1$ then we obtain $N \ge 3$. Observe that in this case $\kp_q(\HH) = \mrq(\G_\HH)=3$ (see~\cite[Claim A.1]{Alon}), and thus the bound is tight.
\end{remark}

\section{Syndrome decoding}
\label{subsec:decoding}

Consider the $(\delta, \HH)$-ECIC based on a matrix $\bL$. 
Suppose that the receiver $R_i$, $i \in [m]$, receives the vector  
\begin{equation}
\by_i = \bx \bL + \bep_i \; ,
\label{eq:error-i}
\end{equation}
where $\bx \bL$ is the codeword transmitted by $S$, and $\bep_i$
is the error pattern affecting this codeword. 

In the classical coding theory, the transmitted vector $\bc$, the received vector $\by$, and the error pattern $\be$ are
related by $\by = \bc + \be$. Therefore, if $\by$ is known to the receiver, then there is a one-to-one correspondence  
between the values of unknown vectors $\bc$ and $\be$. For index coding, however, this is no longer the case. 
The following theorem shows that, in order to recover 
the message $x_{f(i)}$ from $\by_i$ using~(\ref{eq:error-i}), it is sufficient to find just one 
vector from a set of possible error patterns. This set is defined as follows: 
\[
\LL_i(\bep_i) = \left\{ \bep_i + \bz \; : \; \bz \in \spn(\{\bL_j\}_{j \in \Y_i}) \right\} \; . 
\] 
We henceforth refer to the set $\LL_i(\bep_i)$ as the \emph{set of relevant error patterns}. 

\vskip 10pt
\begin{lemma}
\label{thm:relevant_e}
Assume that the receiver $R_i$ receives $\by_i$.
\begin{enumerate}
\item
If $R_i$ knows the message $x_{f(i)}$ then it is able to 
determine the set $\LL_i(\bep_i)$. 
\item
If $R_i$ knows some vector $\beph \in \LL_i(\bep_i)$ then it is able to 
determine $x_{f(i)}$. 
\end{enumerate}
\end{lemma}
\begin{proof}
\begin{enumerate}
\item
From~(\ref{eq:error-i}), we have
\begin{equation} 
\label{2:E10}
\by_i = x_{f(i)} \bL_{f(i)} + \bx_{\X_i}\bL_{\X_i} + \bx_{\Y_i}\bL_{\Y_i} +\bep_i \; .
\end{equation} 
If $R_i$ knows $x_{f(i)}$, then it is also able to determine
\[
\bep_i + \bx_{\Y_i}\bL_{\Y_i} = \by_i - x_{f(i)} \bL_{f(i)} - \bx_{\X_i}\bL_{\X_i} \in \LL_i(\bep_i) \; .
\]
Since $R_i$ has a knowledge of $\bL$, it is also able to determine the whole $\LL_i(\bep_i)$. 

\item
Suppose that $R_i$ knows a vector
\[
\beph= \bep_i + \sum_{j \in \Y_i} z_j \bL_j \in \LL_i(\bep_i) \; ,
\] 
for some $\bz = (z_j)_{\Y_i} \in \fq^{|\Y_i|}$. We show that $R_i$ is able then to determine $x_{f(i)}$.  
Indeed, we re-write (\ref{2:E10}) as
\begin{equation}
\by_i = x_{f(i)}\bL_{f(i)} + \bx_{\X_i}\bL_{\X_i} + (\bx_{\Y_i}- \bz)\bL_{\Y_i} + \beph\; . 
\label{eq:y-i-1}
\end{equation}
The receiver $R_i$ can find some solution of the equation 
\begin{equation}
\label{2:E11} 
\by_i = \hat{x}_{f(i)} \bL_{f(i)} + \bx_{\X_i}\bL_{\X_i} + \bxh_{\Y_i} \bL_{\Y_i}  + \beph\;, 
\end{equation}
with respect to the unknowns $\hat{x}_{f(i)}$ and $\bxh_{\Y_i}$. 
Observe that~(\ref{2:E11}) has at least one solution due to~(\ref{eq:y-i-1}).

From~(\ref{eq:y-i-1}) and~(\ref{2:E11}), we deduce that 
\[
\bO = (\hat{x}_{f(i)} - x_{f(i)})\bL_{f(i)} + (\bxh_{\Y_i} - \bx_{\Y_i} + \bz) \bL_{\Y_i} \; .
\]
This equality implies that $\hat{x}_{f(i)} = x_{f(i)}$ (otherwise, by Corollary~\ref{cor:decodability}, the sum in
the right-hand side will have nonzero weight). 
Hence, $R_i$ is able to determine $x_{f(i)}$, as claimed.  
\end{enumerate}
\end{proof}
\vskip 10pt 

We now describe a syndrome decoding algorithm for linear error-correcting index codes.
From (\ref{2:E10}), we have
\[
\by_i - \bx_{\X_i}\bL_{\X_i} - \bep_i \in \spn\big(\{\bL_{f(i)}\} \cup \{\bL_j\}_{j \in \Y_i}\big) \; .
\]
Let $\cC_i = \spn(\{\bL_{f(i)}\} \cup \{\bL_j\}_{j \in \Y_i})$, and let $\bHi$ be a parity check matrix 
of $\cC_i$. We obtain that  
\begin{equation} 
\bHi \bep_i^T = \bHi(\by_i - \bx_{\X_i}\bL_{\X_i})^T \; .
\label{eq:beta-i}
\end{equation} 
Let $\bbti$ be a column vector defined by
\begin{equation} 
\label{2:E12}
\bbti = \bHi(\by_i - \bx_{\X_i}\bL_{\X_i})^T \; .
\end{equation} 
Observe that each $R_i$ is capable of determining $\bbti$. 
Then we can re-write~(\ref{eq:beta-i}) as
\[
\bHi \bep_i^T = \bbti \; . 
\]
This leads us to the formulation of the following decoding procedure for $R_i$. 

{
\begin{figure}[htb]
\hrule
\vspace{2ex}
\begin{itemize}
  \item {\it Input:} $\by_i$, $\bx_{\X_i}$, $\bL$. 
  \vspace{1ex}
	\item {\it Step 1}: Compute the syndrome
	\[
	\bbti = \bHi(\by_i - \bx_{\X_i}\bL_{\X_i})^T \; . 
	\]  
  \item {\it Step 2}: Find the lowest Hamming weight solution $\hat{\bep}$ of the system
\begin{equation}
\label{2:E13}
\bHi\hat{\bep}^T = \bbti \; .
\end{equation} 
   \item {\it Step 3}: Given that $\bxh_{\X_i} = \bx_{\X_i}$, solve the system for $\hat{x}_{f(i)}$: 
	    \begin{equation} 
			\label{2:E15}
			\by_i = \hat{\bx} \bL + \hat{\bep}.
			\end{equation} 
   \item {\it Output:} $\hat{x}_{f(i)}$.
\end{itemize}
\vspace{2ex}
\hrule
\vspace{1ex}
\caption{Syndrome decoding procedure.}
\label{fig:decoder}
\end{figure}
}
\vskip 40pt

\begin{remark}
Gaussian elimination can be used to solve~(\ref{2:E15}) for $\hat{x}_{f(i)}$. 
However, since $\bL$ also corresponds to an {\mic}, there is more efficient way to do so. From Lemma~\ref{lem:recovery},
there exists a vector $\bv_i \lhd \X_i$ satisfying 
$\bv_i + \be_{f(i)} \in \clspn(\bL)$. Hence $\bv_i + \be_{f(i)} = \bu \bL^T$ for some $\bu \in \fq^N$.
Therefore 
\[
\begin{split} 
\hat{x}_{f(i)} &= \hat{\bx}(\bv_i + \be_{f(i)})^T - \hat{\bx}\bv_i^T \\
&= \hat{\bx}\bL \bu^T - \hat{\bx}\bv_i^T\\
&= (\by_i - \hat{\bep})\bu^T - \hat{\bx}\bv_i^T.
\end{split} 
\] 
With the knowledge of $\bL$ and $\bx_{\X_i}$, $R_i$ can determine $\bu$ and $\hat{\bx}\bv_i^T$. 
Therefore, it can also determine $\hat{x}_{f(i)}$. Note that (\ref{2:E15}) may have more than one solution
$\hat{\bx}$ with $\bxh_{\X_i} = \bx_{\X_i}$. However, as shown in the next theorem, if at most $\delta$
errors occur in $\by_i$, then it always holds that $\hat{x}_{f(i)} = x_{f(i)}$.  
\end{remark}

\medskip
\begin{theorem}
Let $\by_i = \bx \bL + \bep_i$ be the vector received by $R_i$, and let $\weight(\bep_i) \le \delta$. 
Assume that the procedure in Figure~\ref{fig:decoder} is applied to $(\by_i, \bx_{\X_i}, \bL)$. 
Then, its output satisfies $\hat{x}_{f(i)} = x_{f(i)}$. 
\end{theorem}
\begin{proof}
By Lemma~\ref{thm:relevant_e}, it is sufficient to prove that $\hat{\bep} \in \LL_i(\bep_i)$. 
Indeed, since
\[
\bHi \bep_i^T = \bHi \hat{\bep}^T = \bbti \; ,
\]
we have
\[
\bHi (\hat{\bep} - \bep_i)^T = \bO \; .
\]
Hence, $\hat{\bep} - \bep_i \in \cC_i$, and therefore,
\begin{equation} 
\label{2:E14}
\hat{\bep} - \bep_i = z_{f(i)} \bL_{f(i)} + \sum_{j \in \Y_i}z_j\bL_j,
\end{equation}
for some $z_{f(i)} \in \fq$ and $z_j \in \fq$, $j \in \Y_i$. 

Since $\bep_i$ is a solution of~(\ref{2:E13}), and $\weight(\bep_i) \leq \delta$, we 
deduce that $\weight(\hat{\bep}) \leq \delta$ as well. Hence, 
\[
\weight \left( z_{f(i)} \bL_{f(i)} + \sum_{j \in \Y_i} z_j \bL_j \right) = \weight \left(\hat{\bep} - \bep_i \right) \leq 2\delta  \; .
\]
Therefore, by Corollary~\ref {cor:decodability}, $z_{f(i)} = 0$. Hence, $\hat{\bep} \in \LL_i(\bep_i)$, as desired,
and therefore $\hat{x}_{f(i)} = x_{f(i)}$. 
\end{proof}
\vskip 10pt 

\begin{remark}
We anticipate Step 2 in Figure~\ref{fig:decoder} to be computationally hard. 
Indeed, the problem of finding $\hat{\bep}$ over $\ff_2$ of the lowest weight satisfying
\begin{equation} 
\label{EEE3}
\bHi \hat{\bep}^T = \bbti \; ,
\end{equation} 
for a given binary vector $\bbti$ is at least as hard as a decision problem 
{\sc coset weights} that was shown in~\cite{Berlekamp1978} to be NP-complete.
\end{remark}
\vskip 10pt

\section{Static Codes and Related Problems}
\label{sec:StaticECIC}

\subsection{Static Error-Correcting Index Codes}

In the previous sections we focused on linear $\delta$-error-correcting index codes 
for a \emph{particular} instance of the ICSI problem. 
When some of the parameters $m$, $n$, $\X$, and $f$ are variable or not known, it is very likely that
an error-correcting index code for the instance with particular values of these parameters can not be used for 
the instances with different values of some of these parameters. 
Therefore, it is interesting to design an error-correcting 
index code which will be suitable for a \emph{family} of instances of the ICSI problem. 

\vskip 10pt 
\begin{definition}
\label{def:static}
Let $\Gamma = \{ (m, n, \X, f) \}$ be a set of instances for an ICSI problem. 
A $\delta$-error-correcting index code over $\fq$ is said to be \emph{static}
under the set $\Gamma$ if it is a $\delta$-error-correcting $(m,n,\X,f)$-IC over 
$\fq$ for all instances $(m,n,\X,f) \in \Gamma$. 
\end{definition} 
\vskip 10pt 

Recall that an instance $(m,n,\X,f)$ can be described by the side information hypergraph $
\HH(m,n,\X,f)$. For a set $\Gamma$ of instances $(m,n,\X,f)$, let
\begin{equation} 
\label{fkJdef}
\fkJ(\Ga) \define \bigcup_{(m,n,\X,f) \in \Ga} \J(\HH(m,n,\X,f)),
\end{equation} 
where $\J(\HH(m,n,\X,f))$ is defined as in (\ref{Jdef}).
We also define
\[
n(\Ga) \define \max \{n:\ (m,n,\X,f) \in \Ga\}.
\]

\vskip 10pt 
\begin{lemma}
\label{lem:decodability_static}
The $n(\Ga) \times N$ matrix $\bL$ corresponds to a $\delta$-error-correcting index code
which is static under $\Ga$ if and only if 
\[
\weight\left(\sum_{i \in K}z_i\bL_i\right) \geq 2\delta + 1,
\]
for all $K \in \fkJ(\Ga)$ and for all choices of $z_i \in \fq^*$, $i \in K$. 
\end{lemma}
\begin{proof}
The proof follows from Definition~\ref{def:static} and Lemma~\ref{lem:decodability}. 
\end{proof} 
\medskip
Please notice that when $\bL$ is used for an instance $(m,n,\X,f) \in \Ga$ with 
$n < n(\Ga)$, then the last $n(\Ga) - n$ rows of $\bL$ are simply discarded.  
\vskip 10pt 

One particular family of interest is $\Ga(n,\rho)$, the family that contains 
all instances where each receiver owns at least $n - \rho$ messages as its side information.
More formally, 
\begin{multline*}
\Ga(n,\rho) = \big\{ (m,n',\X,f) \; : \; n' \leq n \\
\mbox{ and } \forall i \in [m], \; |\X_i| \geq n - \rho \big\} \; .
\end{multline*}
A $\delta$-error-correcting index code which is 
static under $\Ga(n,\rho)$ will provide successful communication between the sender
and the receivers under the presence of at most $\delta$ errors, 
despite a possible change of the collection of the side information sets $\X$, 
a change of the set of receivers, and a change of the demand function, 
as long as each receiver still possesses at least $n - \rho$ messages. 
\vskip 10pt
In the rest of this section, we assume that $N \geq 1$, $n \geq \rho \geq 1$ and $\delta \geq 0$. 
\medskip
\begin{definition}
\label{def:property}
An $n \times N$ matrix $\bL$ is said to satisfy the $(\rho,\delta)$-Property
if any nontrivial linear combination of at most $\rho$ rows of $\bL$ has weight at least $2\delta + 1$.  
\end{definition}
\vskip 10pt

\begin{proposition}
\label{pro:static}
The $n \times N$ matrix $\bL$ corresponds to a $\delta$-error-correcting linear index code, 
which is static under $\Gnr$, if and only if $\bL$ satisfies the $(\rho,\delta)$-Property. 
\end{proposition}
\begin{proof}
Let $\bL$ be an $n \times N$ matrix that satisfies the $(\rho,\delta)$-Property.
We show that this is equivalent to the condition that 
$\bL$ corresponds to a $\delta$-error-correcting linear index code, 
which is static under $\Gnr$. 
By Lemma~\ref{lem:decodability_static}, it suffices to show that $\fkJ(\Ga(n,\rho))$ is the collection
of all nonempty subsets of $[n]$, whose cardinalities are not greater than $\rho$. 

Consider an instance $(m,n',\X,f) \in \Ga(n,\rho)$. For all $i \in [m]$, we have $|\X_i| \geq n - \rho$ and 
$\Y_i = [n']\backslash (f(i) \cup \X_i)$, and thus we deduce that
\[
|\Y_i| \leq n' - 1 - (n - \rho) \leq n' - 1 - (n' - \rho) = \rho - 1.    
\]
Hence by (\ref{Jdef}), the cardinality of each set in $\J(\HH(m,n',\X,f))$
is at most
\[
1 + (\rho - 1) = \rho.
\]
Therefore, due to (\ref{fkJdef}), every set in $\fkJ(\Ga(n,\rho))$ has at most $\rho$ elements. 

It remains to show that every nonempty subset of $[n]$ whose cardinality is at most $\rho$ 
belongs to $\fkJ(\Ga(n,\rho))$. 
Consider an arbitrary $\rho'$-subset $K=\{i_1,i_2,\ldots,i_{\rho'}\}$ of $[n]$, with $1 \leq \rho' \leq \rho$.
Consider an instance
$(m=1,n,\X,f) \in \Ga(n,\rho)$ with $\X_1 = [n] \backslash K$ and $f(1) = i_1$. 
Since
\[
\Y_1 = K \backslash \{i_1\}, 
\]
we have
\[
K = \{i_1\} \cup \Y_1 \in \J(\HH(m,n,\X,f)) \subseteq \fkJ(\Ga(n,\rho)).
\]
The proof follows. 
\end{proof} 
\vskip 10pt

\subsection{Application: Weakly Resilient Functions}

In this section we introduce the notion of weakly resilient functions. 
Hereafter, we restrict the discussion to the binary alphabet.
 
The concept of \emph{binary resilient functions} was first introduced by Chor {\et} in \cite{Chor1985}
and independently by Bennet {\et} in \cite{Bennet1988}. 

\vskip 10pt
\begin{definition}
A function $\bff:\ \ft^N \ra \ft^n$ is called \emph{$t$-resi\-li\-ent} if
$\bff$ satisfies the following property: when $t$ arbitrary inputs of $\bff$ are fixed and the remaining 
$N-t$ inputs run through all the $2^{N-t}$-tuples exactly once, 
the value of $\bff$ runs through every possible output $n$-tuple
an equal number of times. 
Moreover, if $\bff$ is a linear
transformation then it is called a \emph{linear $t$-resilient function}. 
We refer to the parameter $t$ as the \emph{resiliency} of $\bff$. 
\end{definition} 
\vskip 10pt 

The applications of resilient functions can be found in fault-tolerant distributed computing, 
quantum cryptographic key distribution~\cite{Chor1985}, 
privacy amplification~\cite{Bennet1988} and random sequence generation
for stream ciphers~\cite{Carlet2010}. 
Connections between linear error-correcting codes and resilient functions
were established in~\cite{Chor1985}.
\medskip
\begin{theorem}[\cite{Chor1985}]
\label{thm:ecc_rf}
Let $\bL$ be an $n \times N$ binary matrix. Then 
$\bL$ is a generator matrix of a linear error-correcting code with minimum distance 
$d = t + 1$ if and only if $\bff(\bz) = \bL \bz^T$ is $t$-resilient.    
\end{theorem}
\vskip 10pt 

\begin{remark}
Vectorial boolean functions with certain properties are useful for design of stream ciphers. 
These properties include high resiliency and high nonlinearity (see, for instance, \cite{Carlet2010}). 
However, linear resilient functions are still particularly interesting, since they 
can be transformed into highly nonlinear resilient functions with the same parameters.
This can be achieved by a composition of the linear function with a highly nonlinear 
permutation (see~\cite{ZhangZheng1995,GuptaSarkar2002} for more details).  
\end{remark}
\medskip

Below we introduce a definition of a $\rho$-weakly $t$-resilient function, 
which is a weaker version of a $t$-resilient function. 

\vskip 10pt 
\begin{definition}
A function $\bff:\ \ft^N \ra \ft^n$ is called \emph{$\rho$-weakly $t$-resilient} if
$\bff$ satisfies the property that every set of $\rho$ coordinates in the image of $\bff$ 
runs through every possible output $\rho$-tuple
an equal number of times, when $t$ arbitrary inputs of $\bff$ are fixed and the remaining 
$N-t$ inputs run through all the $2^{N-t}$-tuples exactly once. 
\end{definition}
\vskip 10pt 

\begin{remark}
A $\rho$-weakly $t$-resilient function $\bff:\ \ft^N \ra \ft^n$ can be viewed as a collection of 
$\binom{n}{\rho}$ different $t$-resilient functions $\ft^N \ra \ft^\rho$, each such function is obtained by taking 
some $\rho$ coordinates in the image of $\bff$. 
Similarly to~\cite{Chor1985}, consider a scenario, in which two parties are sharing a
secret key, which consists of $N$ randomly selected bits. Suppose that at some moment 
$t$ out of the $N$ bits of the key are leaked to an adversary. By applying 
a $t$-resilient function to the current $N$-bit key, two parties are able to obtain a completely new
and secret key of $n$ bits, without requiring any communication or randomness generation. 
However, if the parties use various parts of the key for various purposes, 
they may only require one of the $\rho$-bit secret keys (instead of the larger $n$-bit key). In that case 
a $\rho$-weakly $t$-resilient function can be used. By applying a $\rho$-weakly
$t$-resilient function to the current $N$-bit key, the parties obtain a set of $\binom{n}{\rho}$ different
$\rho$-bit keys, each key is new and secret (however these keys might not be independent of each other). 
\end{remark}
\vskip 10pt 


\begin{theorem}
Let $\bL$ be an $n \times N$ binary matrix. 
Then $\bL$ satisfies the $(\rho,\delta)$-Property if and only if the function 
$\bff : \ft^N \rightarrow \ft^n$ defined by $\bff(\bz) = \bL \bz^T$
is $\rho$-weakly $2\delta$-resilient. 
\end{theorem}
\begin{proof}
\begin{enumerate}
\item
Suppose that $\bL$ satisfies the $(\rho,\delta)$-Property. 
Take any $\rho$-subset $K \subseteq [n]$. By Definition~\ref{def:property}, 
the $\rho \times N$ submatrix $\bL_K$ of $\bL$ is a generating matrix of 
the error-correcting code with the minimum distance $\ge 2 \delta + 1$. By Theorem~\ref{thm:ecc_rf}, 
the function $\bff_K : \ft^N \rightarrow \ft^\rho$ defined by 
$\bff_K(\bz) = \bL_K \bz^T$ is $2 \delta$-resilient. 
Since $K$ is an arbitrary $\rho$-subset of $[n]$, the function $\bff$ is $\rho$-weakly $2 \delta$-resilient.   
\item
Conversely, assume that the function $\bff$ is $\rho$-weakly $2 \delta$-resilient. 
Take any subset $K \subseteq [n]$, $|K| = \rho$. 
Then the function $\bff_K : \ft^N \rightarrow \ft^\rho$ defined by $\bff_K(\bz) = \bL_K \bz^T$ is $2 \delta$-resilient.
Therefore, by Theorem~\ref{thm:ecc_rf}, 
$\bL_K$ is a generating matrix of a linear code with minimum distance $2 \delta + 1$. Since $K$ is
an arbitrary $\rho$-subset of $[n]$, by Proposition~\ref{pro:static} $\bL$ satisfies the $(\rho,\delta)$-Property. 
\end{enumerate}
\end{proof}

\subsection{Bounds and Constructions}

In this section we study the problem of constructing a matrix $\bL$ satisfying the $(\rho,\delta)$-Property. 
Such $\bL$ with the minimal possible number of columns is called \emph{optimal}.
First, observe that from Proposition~\ref{pro:static} we have 
\[
\fkJ(\Ga(n,\rho)) = \bigcup_{i = 1}^\rho \binom{[n]}{i},
\]
is the set of all nonempty subsets of $[n]$ of cardinality at most $\rho$.  
Next, consider an instance $(m^*,n,\X^*,f^*)$ satisfying 
\begin{equation}
\label{eq:e30} 
\J(\Hs) = \fkJ(\Ga(n,\rho)) \; ,
\end{equation} 
where $\Hs = \HH(\issn)$ is the side information hypergraph corresponding to that instance. 
Such an instance can be constructed as follows. For each subset $K = \{ i_1, i_2, \ldots, i_{\rho'} \} 
\subseteq [n]$ ($1 \leq \rho' \leq \rho)$, we introduce a receiver which requests the message $x_{i_1}$, 
and has a set $\{x_j: \ j \in [n] \backslash K\}$ as its side information. 
It is straightforward to verify that indeed we obtain an instance $\iss$ satisfying
(\ref{eq:e30}). The problem of designing an optimal matrix $\bL$ satisfying the $(\rho,\delta)$-Property
then becomes equivalent to the problem of finding an optimal $(\delta,\Hs)$-ECIC. 
Thus, $\N_q[\Hs,\delta]$ is equal to the number of columns in an optimal
matrix which satisfies the $(\rho,\delta)$-Property.  

The corresponding $\al$-bound and $\kp$-bound for $\N_q[\Hs, \delta]$ can be stated as follows. 

\vskip 10pt 
\begin{theorem}
\label{thm:bounds_static}
Let $\rho^*$ be the smallest number such that a linear $[n,n-\rho^*,\geq \rho+1]_q$ code exists. 
Then we have
\[
N_q[\rho,2\delta+1] \leq \N_q[\Hs, \delta] \leq N_q[\rho^*,2\delta+1] \; . 
\]
\end{theorem}
\begin{proof}
The first inequality follows from the $\al$-bound and from the fact that $\al(\Hs) = \rho$, 
which is due to (\ref{eq:e30}).

For the second inequality, it suffices to show that $\kp_q(\Hs) = \rho^*$. 
By Corollary~\ref{coro:ic_decodability}, an $n \times N$ matrix $\bL$ corresponds to an $\Hs$-IC
if and only if $\{\bL_i: \ i \in K\}$ is linearly independent for every $K \in \J(\Hs)$. 
Since $\J(\Hs)$ is the set of all nonempty subsets of cardinality at most $\rho$, this
is equivalent to saying that every set of at most $\rho$ rows of $\bL$ is linearly independent. 
This condition is equivalent to the condition that $\bL^T$ is a
parity check matrix of a linear code with the minimum distance at least $\rho+1$~\cite[Chapter 1]{MW_and_S}. 
Therefore, a linear $\Hs$-IC of length $N$ exists if and only if an $[n, n - N, \ge \rho+1]_q$ linear code exists. 
Since $\rho^*$ is the smallest number such that an $[n,n - \rho^*, \geq \rho+1]_q$ code exists, we conclude 
that $\kp_q(\Hs) = \rho^*$. 
\end{proof} 
\vskip 10pt 

\begin{corollary}
\label{coro:singleton_static}
The length of an optimal $\delta$-error-correcting linear index code over $\fq$
which is static under $\Ga(n,\rho)$ satisfies
\[ 
\N_q[\delta,\Hs] \geq \rho^* + 2\delta,
\]
where $\rho^*$ is the smallest number such that an $[n,n-\rho^*,\geq \rho+1]_q$ code exists.
\end{corollary}
\begin{proof} 
This is a straightforward corollary of Theorem~\ref{thrm:singleton} (the Singleton bound)
and Theorem~\ref{thm:bounds_static}. 
\end{proof} 
\vskip 10pt 
 
\begin{corollary} 
\label{coro:optimal_staticIC}
For $q \geq \max\{n-1,\rho + 2\delta - 1\}$, the length of an optimal $\delta$-error-correcting linear index code over $\fq$
which is static under $\Ga(n,\rho)$ is $\rho + 2\delta$. 
\end{corollary} 
\begin{proof}
For $q \geq n - 1$ there exists an $[n, n-\rho^*, \rho + 1]_q$ linear code with $\rho^* = \rho$ 
(for example, one can take an extended RS code~\cite[Chapter 11]{MW_and_S}). 
Due to Singleton bound, we conclude that $\rho^* = \rho$ is the smallest value such that 
$[n, n-\rho^*, \rho + 1]_q$ linear code exists. 
Following the lines of the proof of Theorem~\ref{thm:bounds_static}, there exists a $\delta$-error-correcting index code
of length $N_q[\rho, 2\delta + 1]$, which is static under $\Gnr$. As $q \geq \rho + 2\delta - 1$, we have 
\[
N_q[\rho, 2\delta + 1] = \rho + 2\delta \;  
\]
(for example, by taking an extended RS code). 
Due to Corollary~\ref{coro:singleton_static}, this static error-correcting index code is optimal.   
\end{proof} 
\vskip 10pt 

\begin{remark}
\label{rm:staticIC_ecc}
We observe from the proof of Theorem~\ref{thm:bounds_static} that the problem of constructing
an optimal linear (non-error-correcting) index code, which is static under $\Ga(n,\rho)$, is, in fact,
equivalent to the problem of constructing a parity check matrix of a classical linear error-correcting code. 
\end{remark} 

\vskip 10pt 
\begin{example} 
\label{ex:static_IC}
Let $n = 20$, $\rho = 10$, $\delta = 1$ and $q = 2$. 
From \cite{Grassl}, the smallest possible dimension of a binary linear code of length $20$ and 
minimum distance $11$ is $3$. We obtain that $\rho^* = 17$. We also have $N_2[17,3] = 22$.  
Theorem~\ref{thm:bounds_static} implies the existence of a
one-error-correcting binary index code of length $22$ which can be
used for any instance of IC problem, in which each receiver owns at least $10$
out of (at most) $20$ messages, as side information. It also implies that the length of any such static
error-correcting index code is at least $N_2[10,3] = 14$. 
Corollary~\ref{coro:singleton_static} provides a better lower bound on the minimum length, which is $17 + 2 = 19$. 
\end{example}
\vskip 10pt

\begin{example}
\label{ex:static}
Below we show that with the same number of inputs $N$ and outputs $n$, 
a weakly resilient function may have strictly higher resiliency $t$. 
From Example~\ref{ex:static_IC}, there exists a linear vectorial Boolean function $\bff: (\ft)^{22} \ra (\ft)^{20}$
which is $10$-weakly $2$-resilient. According to \cite{Grassl}, an optimal linear $[22,20]_2$
code has minimum distance $d = 2$. Hence, due to Theorem~\ref{thm:ecc_rf}, the resiliency of any linear
vectorial Boolean function $\bg: (\ft)^{22} \ra (\ft)^{20}$ cannot exceed one.   
\end{example} 
\vskip 10pt 

The problem of constructing an $n \times N$ matrix $\bL$ which satisfies the $(\rho,\delta)$-Property is a natural generalization of the problem of constructing the parity check matrix $\bH$ of a linear $[n, k, d \geq \rho + 1]_q$ code. Indeed, $\bH$ is a parity check matrix of an $[n, k, d \geq \rho + 1]_q$ code if and only if every set of $\rho$ columns of $\bH$ is linearly independent. Equivalently, any nontrivial linear combination of at most $\rho$ columns of $\bH$ has weight at least one. For comparison, $\bL$ satisfies the $(\rho,\delta)$-Property if and only if any nontrivial linear combination of at most $\rho$ columns of $\bL^T$ has weight at least $2\delta + 1$.  

Some classical methods for deriving bounds on the parameters of error-correcting codes can be generalized to 
the case of linear static error-correcting index codes. Below we present a Gilbert-Varshamov-like bound.

\vskip 10pt 
\begin{theorem}
\label{thm:GVbound}
Let $V_q(N, 2\delta)$ denotes the volume of $q$-ary sphere of radius $2 \delta$ in $\fq^N$ 
given by~(\ref{eq:sphere}). If 
\[
\sum_{i = 0}^{\rho-1}\binom{n - 1}{i} (q-1)^i  < \frac{q^N}{V_q(N, 2\delta)} \; ,
\]
then there exists an $n \times N$ matrix $\bL$ which satisfies the $(\rho,\delta)$-Property. 
\end{theorem}
\begin{proof} 
We build up the set $\R$ of rows of $\bL$ one by one.
The first row can be any vector in $\fq^N$ of weight at least $2\delta+1$.  
Now suppose we have chosen $r$ rows so that no nontrivial linear combination of at most $\rho$ among these $r$ rows have weight less than $2\delta + 1$. There are at most
\[
V_q(N, 2\delta) \; \sum_{i = 0}^{\rho-1} \binom{r}{i} (q-1)^i 
\]
vectors which are at distance less than $2\delta+1$ from any linear combination of at most $\rho-1$ among $r$ chosen rows
(this includes vectors at distance less than $2\delta+1$ from $\bO$). 
If this quantity is smaller than $q^N$, then we can add another row to the set
$\R$ so that no nontrivial linear combination of at most $\rho$ rows in $\R$ has weight less than $2\delta + 1$. 
The claim follows if we replace $r$ by $n-1$. 
\end{proof} 
\vskip 10pt 

\begin{remark} 
If we apply Theorem~\ref{thrm:random} to the instance $(m^*, n, \X^*, f^*)$ defined in the beginning of this section, 
then we obtain a bound, which is somewhat weaker then its counterpart in Theorem~\ref{thm:GVbound}, namely
the $n \times N$ matrix $\bL$ as above exists if 
\[
\sum_{i = 1}^{\rho} \binom{n}{i} q^{i-1}  < \frac{q^N}{V_q(N, 2\delta)} \; . 
\]
\end{remark}

\section{Conclusions}
\label{sec:conclusion}

In this work, we generalize Index Coding with Side Information problem towards a setup with errors. 
Under this setup, each receiver should be able to recover its desired message even 
if a certain amount of errors happen in the transmitted data. This is the first work that 
considers such a problem. 

A number of bounds on the length of an optimal error-correcting index code are constructed. 
As it is shown in Example~\ref{ex:1}, a separation of error-correcting code and index code
sometimes leads to a non-optimal scheme. This raises a question of designing
coding schemes in which the two layers are treated as a whole. 
Therefore, the question of constructing error-correcting index codes with good 
parameters is still open.

A general decoding 
procedure for linear error-correcting index codes is discussed. The difference between decoding of
a classical error-correcting code and decoding of an error-correcting index code is that in the latter case, 
each receiver does not require a complete knowledge of the error vector. This difference
may help to ease the decoding process. Finding an efficient decoding method for error-correcting index codes (together 
with their corresponding constructions) is also still an open problem. 

The notion of error-correcting index code is further generalized to static index code. The latter is designed 
to serve a family of instances of error-correcting index coding problem. The problem of designing
an optimal static ECIC is studied, and several bounds on the length of such codes are presented.

\section{Acknowledgements}

The authors would like to thank the authors of~\cite{Yossef-journal} for providing a preprint of their paper. 
This work is supported by the National Research Foundation of Singapore (Research Grant
NRF-CRP2-2007-03).

\appendix
\label{app:A}

\begin{lemma}
If $\G_\HH$ is symmetric, then the generalized independence number of $\HH$ is the 
independence number of $\G_\HH$.  
\end{lemma}
\begin{proof}
It suffices to show that if $\G_\HH$ is symmetric, then the set of generalized independent sets of $\HH$ and 
the set of independent sets of $\G_\HH$ coincide. 

Let $H$ be a generalized independent set in $\HH$. If $|H| = 1$, then obviously $H$ is an independent set in $\G_\HH$. 
Assume that $|H| \ge 2$. For any pair of vertices $i,j \in H$, the set $\{i,j\}$
belongs to $\JX$. By definition of $\JX$, either there is no edge from $i$ to $j$, or there is no 
edge from $j$ to $i$, in $\G_\HH$. Since $\G_\HH$ is symmetric, there are no edges between $i$ and $j$, in neither
directions. Therefore, $H$ is an independent set in $\G_\HH$. 

Conversely, let $H$ be an independent set in $\G_\HH$. For each $i \in H$, since there are no edges from $i$
to all other vertices in $H$, we deduce that $H \backslash \{i\} \seq \Y_i$. Due to~(\ref{Jdef}), every 
subset of $H$ which contains $i$ belongs to $\JX$. This holds for an arbitrary $i \in H$. Therefore,
every nonempty subset of $H$ belong to $\JX$. We obtain that $H$ is a generalized independent set of $\HH$.     
\end{proof}

\bibliographystyle{IEEEtran}
\bibliography{IndexCoding_ErrorCorrection}

\begin{thebibliography}{10}
\providecommand{\url}[1]{#1}
\csname url@samestyle\endcsname
\providecommand{\newblock}{\relax}
\providecommand{\bibinfo}[2]{#2}
\providecommand{\BIBentrySTDinterwordspacing}{\spaceskip=0pt\relax}
\providecommand{\BIBentryALTinterwordstretchfactor}{4}
\providecommand{\BIBentryALTinterwordspacing}{\spaceskip=\fontdimen2\font plus
\BIBentryALTinterwordstretchfactor\fontdimen3\font minus
  \fontdimen4\font\relax}
\providecommand{\BIBforeignlanguage}[2]{{%
\expandafter\ifx\csname l@#1\endcsname\relax
\typeout{** WARNING: IEEEtran.bst: No hyphenation pattern has been}%
\typeout{** loaded for the language `#1'. Using the pattern for}%
\typeout{** the default language instead.}%
\else
\language=\csname l@#1\endcsname
\fi
#2}}
\providecommand{\BIBdecl}{\relax}
\BIBdecl

\bibitem{BirkKol98}
Y.~Birk and T.~Kol, ``Informed-source coding-on-demand (\upshape{ISCOD}) over
  broadcast channels,'' in \emph{Proc. IEEE Conf. on Comput. Commun.
  (INFOCOM)}, San Francisco, CA, 1998, pp. 1257--1264.

\bibitem{BirkKol2006}
------, ``Coding-on-demand by an informed source (\upshape{ISCOD}) for
  efficient broadcast of different supplemental data to caching clients,''
  \emph{IEEE Trans. Inform. Theory}, vol.~52, no.~6, pp. 2825--2830, 2006.

\bibitem{Rouayheb2009}
S.~{El Rouayheb}, A.~Sprintson, and C.~Georghiades, ``On the index coding
  problem and its relation to network coding and matroid theory,'' submitted to
  \emph{IEEE Trans. Inform. Theory}.

\bibitem{Katti2006}
S.~Katti, H.~Rahul, W.~Hu, D.~Katabi, M.~M\'edard, and J.~Crowcroft, ``Xors in
  the air: Practical wireless network coding,'' in \emph{Proc. ACM SIGCOMM},
  2006, pp. 243--254.

\bibitem{Katti2008}
S.~Katti, D.~Katabi, H.~Balakrishnan, and M.~M\'edard, ``Symbol-level network
  coding for wireless mesh networks,'' \emph{ACM SIGCOMM Comput. Commun.
  Review}, vol.~38, no.~4, pp. 401--412, 2008.

\bibitem{Yossef}
Z.~Bar-Yossef, Z.~Birk, T.~S. Jayram, and T.~Kol, ``Index coding with side
  information,'' in \emph{Proc. 47th Annu. IEEE Symp. on Found. of Comput. Sci.
  (FOCS)}, 2006, pp. 197--206.

\bibitem{Yossef-journal}
------, ``Index coding with side information,'' \emph{IEEE Trans. Inform.
  Theory}, to appear.

\bibitem{LubetzkyStav}
E.~Lubetzky and U.~Stav, ``Non-linear index coding outperforming the linear
  optimum,'' \emph{Proc. 48th Annu. IEEE Symp. on Found. of Comput. Sci.
  (FOCS)}, pp. 161--168, 2007.

\bibitem{Wu}
Y.~Wu, J.~Padhye, R.~Chandra, V.~Padmanabhan, and P.~A. Chou, ``The local
  mixing problem,'' in \emph{Proc. Inform. Theory and Applicat. Workshop}, San
  Diego, CA, 2006.

\bibitem{Rouayheb2007}
S.~{El Rouayheb}, M.~A.~R. Chaudhry, and A.~Sprintson, ``On the minimum number
  of transmissions in single-hop wireless coding networks,'' in \emph{Proc.
  IEEE Inform. Theory Workshop (ITW)}, 2007, pp. 120--125.

\bibitem{Rouayheb2008}
S.~{El Rouayheb}, A.~Sprintson, and C.~Georghiades, ``On the relation between
  the index coding and the network coding problems,'' in \emph{Proc. IEEE Symp.
  on Inform. Theory (ISIT)}, Toronto, Canada, 2008, pp. 1823--1827.

\bibitem{ChaudhrySprintson}
M.~A.~R. Chaudhry and A.~Sprintson, ``Efficient algorithms for index coding,''
  in \emph{Proc. IEEE Conf. on Comput. Commun. (INFOCOM)}, 2008, pp. 1--4.

\bibitem{Alon}
N.~Alon, A.~Hassidim, E.~Lubetzky, U.~Stav, and A.~Weinstein, ``Broadcasting
  with side information,'' in \emph{Proc. 49th Annu. IEEE Symp. on Found. of
  Comput. Sci. (FOCS)}, 2008, pp. 823--832.

\bibitem{Ahlswede}
R.~Ahlswede, N.~Cai, S.~Y.~R. Li, and R.~W. Yeung, ``Network information
  flow,'' \emph{IEEE Trans. Inform. Theory}, vol.~46, pp. 1204--1216, 2000.

\bibitem{KoetterMedard2003}
R.~Koetter and M.~M\'{e}dard, ``An algebraic approach to network coding,''
  \emph{IEEE/ACM Trans. Netw.}, vol.~11, pp. 782--795, 2003.

\bibitem{Haemers1978}
W.~Haemers, ``An upper bound for the shannon capacity of a graph,''
  \emph{Algebr. Methods Graph Theory}, vol.~25, pp. 267--272, 1978.

\bibitem{DauSkachekChee2010}
S.~H. Dau, V.~Skachek, and Y.~M. Chee, ``On the security of index coding with
  side information,'' submitted. Also available online at {\tt
  http://arxiv.org/abs/1102.2797.}

\bibitem{Grassl}
M.~Grassl, ``Bounds on the minimum distance of linear codes and quantum
  codes,'' available online at {\tt http://www.codetables.de}.

\bibitem{Chudnovsky}
M.~Chudnovsky, N.~Robertson, P.~Seymour, and R.~Thomas, ``The strong perfect
  graph theorem,'' \emph{Annals of Mathematics}, vol. 164, pp. 51--229, 2006.

\bibitem{Berlekamp1978}
E.~R. Berlekamp, R.~J. McEliece, and H.~C.~A. van Tilborg, ``On the inherent
  intractability of certain coding problems,'' \emph{IEEE Trans. Inform.
  Theory}, vol. IT-24, no.~3, pp. 384--386, 1978.

\bibitem{Chor1985}
B.~Chor, O.~Goldreich, J.~H\r{a}stad, J.~Freidmann, S.~Rudich, and
  R.~Smolensky, ``The bit extraction problem or t-resilient functions,'' in
  \emph{Proc. 26th Annu. IEEE Symp. on Found. of Comput. Sci. (FOCS)}, 1985,
  pp. 396--407.

\bibitem{Bennet1988}
C.~H. Bennet, G.~Brassard, and J.~M. Robert, ``Privacy amplification by public
  discussion,'' \emph{SIAM J. Computing}, vol.~17, pp. 210--229, 1988.

\bibitem{Carlet2010}
C.~Carlet, \emph{Vectorial Boolean Functions for Cryptography}, ser. Boolean
  Models and Methods in Mathematics, Computer Science and Engineering.\hskip
  1em plus 0.5em minus 0.4em\relax Cambridge University Press, 2010, ch.~9.

\bibitem{ZhangZheng1995}
X.-M. Zhang and Y.~Zheng, ``On nonlinear resilient functions,'' in \emph{Proc.
  14th Annu. Int. Conf. on Theory and Appl. of Cryptographic Tech.
  (EUROCRYPT)}, 1995, pp. 274--288.

\bibitem{GuptaSarkar2002}
K.~Gupta and P.~Sarkar, ``Improved construction of nonlinear resilient
  s-boxes,'' \emph{IEEE Trans. Inform. Theory}, vol.~51, no.~1, pp. 339--348,
  2005.

\bibitem{MW_and_S}
F.~J. MacWilliams and N.~J.~A. Sloane, \emph{The Theory of Error-Correcting
  Codes}.\hskip 1em plus 0.5em minus 0.4em\relax Amsterdam: North-Holland,
  1977.

\end{thebibliography}

\end{document}